\documentclass[11pt,reqno]{amsart}
\usepackage{tikz}
\usepackage{amsthm}
\usepackage{bm}
\usepackage{amsmath}
\usepackage{amsmath,amssymb}
\usepackage{subfig}
\usepackage{epstopdf}
\usepackage{epsfig}
\graphicspath{{./figures/}}
\usepackage{tikz}
\usetikzlibrary{shapes,arrows}
\tikzstyle{decision} = [diamond, draw, fill=blue!20, 
    text width=4.5em, text badly centered, node distance=3cm, inner sep=0pt]
\tikzstyle{block} = [rectangle, draw, fill=blue!20, 
    text width=5em, text centered, rounded corners, minimum height=4em]
\tikzstyle{line} = [draw, -latex']
\tikzstyle{cloud} = [draw, ellipse,fill=red!20, node distance=3cm,
    minimum height=2em]
 \newtheoremstyle{mystyle}{36pt}{}{}{}{\bfseries}{.}{ }{}
 
  \theoremstyle{plain}
  \theoremstyle{plain}
\newtheorem{theorem}{Theorem}[section]
\newtheorem{proposition}[theorem]{Proposition}
\newtheorem{lemma}[theorem]{Lemma}
\newtheorem{corollary}[theorem]{Corollary}

\theoremstyle{definition}
\newtheorem{definition}[theorem]{Definition}
\theoremstyle{remark}
\newtheorem{remark}[theorem]{Remark}
 \newtheorem{example}{Example}

    \usetikzlibrary{positioning}
\tikzset{main node/.style={circle,fill=blue!20,draw,minimum size=1cm,inner sep=0pt},  }

\newcommand{\sP}{\mathcal{P}}

\newcommand{\dd}{\mathcal{\dagger}}
\newcommand{\g}{\mathbf{g}}
\newcommand{\grad}{\mathrm{grad}}
\newcommand{\V}{\mathsf{V}}

\newcommand{\Hess}{\mathrm{Hess}}
\newcommand{\R}{\mathrm{R}}

\newcommand{\K}{\mathrm{K}}

\newcommand{\Dist}{\mathrm{Dist}}
\begin{document}
\title[Curvatures in hydrodynamical density manifolds]{Geometric calculations on density manifolds from reciprocal relations in hydrodynamics}
\author[Li]{Wuchen Li}
\email{wuchen@mailbox.sc.edu}
\address{Department of Mathematics, University of South Carolina, Columbia, 29208.}
\begin{abstract}
Hydrodynamics describes the evolution of macroscopic states in non--equilibrium thermodynamics. Following Onsager reciprocal relations, one can formulate a large class of hydrodynamic equations as gradient flows of free energies. In recent years, such Onsager gradient flows have been extensively investigated on optimal transport type metric spaces with nonlinear mobilities, known as hydrodynamical density manifolds. A typical example is the gradient--drift Fokker--Planck equation, which can be characterized as the gradient flow of the free energy in the Wasserstein-2 metric space. This paper studies geometric calculations on general hydrodynamical density manifolds. We first formulate the associated Levi--Civita connections, gradients, Hessians, and parallel transports, and then derive the corresponding Riemannian and sectional curvatures. Finally, we obtain closed-form formulas for sectional curvatures in one dimensional density manifolds, where the signs of the curvatures are determined by the convexity of the mobility functions. As illustrations, we present density manifolds and their sectional curvatures for several zero range models, including independent particle systems, simple exclusion processes, and Kipnis--Marchioro--Presutti models.
\end{abstract}
\maketitle
\section{Introductions.}
Non--equilibrium thermodynamics plays a central role in modeling complex physical systems, chemical reactions, and biological phenomena \cite{MFT, GEN, Onsager}. For systems out of equilibrium, time--dependent statistical physics models are required. Typical examples are hydrodynamic equations, which describe irreversible macroscopic processes such as the evolution of particle density. Based on Onsager reciprocal relations \cite{Onsager}, irreversible processes often exhibit symmetric dissipative fluctuations toward equilibria. In recent years, Onsager principle has been extended to more general non--symmetric and non--equilibrium diffusive systems. Representative developments include macroscopic fluctuation theory (MFT) \cite{MFT} and the  {general equation for non--equilibrium reversible-irreversible coupling} (GENERIC) framework \cite{GEN}.

A particular form of Onsager reciprocal relation has recently gained attention in the {mathematics community}. It identifies a class of hydrodynamic equations, such as the independent particle zero range model \cite{MFT}, that satisfy a gradient flow structure on density spaces. This structure corresponds to the gradient--drift Fokker--Planck equation, well known in the optimal transport community as a Wasserstein gradient flow \cite{JKO, otto2001}. In fact, this gradient operator defines an infinite dimensional Riemannian manifold whose metric is widely studied as the Wasserstein--2 metric \cite{am2006, otto2001, vil2008}. {In particular, the geometric Riemannian formalism for density manifolds originates from Otto's work \cite{otto2001}, known as Otto calculus, which was initially developed for the study of Wasserstein--$2$ gradient flows.} More generally, the Onsager response operator associated with hydrodynamics induces a family of Wasserstein type Riemannian metrics with nonlinear mobilities \cite{C1, LiYing}. In this paper, we refer to density spaces equipped with such nonlinear Wasserstein-type metrics as hydrodynamical density manifolds.

Geometric calculations on hydrodynamical density manifolds have become essential for understanding fluctuation relations and convergence properties in general non--equilibrium thermodynamics. For example, \cite{PHC, Dean, WD} introduces Langevin dynamics on density spaces--often called super-Brownian motion or stochastic Fokker--Planck equation whose behavior depends on geometric quantities such as second--order differential operators on these manifolds. This raises a natural question: {\em What are the geometric quantities, such as the Riemannian and sectional curvatures, of density manifolds?} 

In this paper, we focus on the first question. We derive several geometric calculations for hydrodynamical density manifolds with nonlinear mobilities, {which can be viewed as generalized Otto calculuses \cite{otto2001}.} We first study the Levi--Civita connection, gradient, and Hessian operators, and subsequently compute the Riemannian curvature and sectional curvature tensors. In one dimensional density spaces, we obtain explicit formulas for these curvatures in Theorem \ref{thm3} and show that the convexity of the mobility function determines the sign of the sectional curvature. We further present explicit geometric computations for several zero range models, including independent particles, simple exclusion processes, and the Kipnis--Marchioro--Presutti model \cite{KMP}.

In the literature, geometric computations on the Wasserstein--2 density manifold were first introduced by \cite{Lafferty} in Lagrangian coordinates and later derived by \cite{Lott} in Eulerian coordinates. The second order analysis in the Wasserstein-2 space has been studied in \cite{Gigli2011}. These correspond to two classical descriptions in fluid mechanics: the flow map formulation and the density based formulation. Hessian operators in generalized density manifolds have been investigated in \cite{C1, LiG1, LiG2, LiG3, LiG4, LiG5}. However, systematic Riemannian computations--particularly curvature formulas--for general hydrodynamical density manifolds remain largely unknown. Here we develop Riemannian calculations for manifolds with nonlinear mobility functions using the Eulerian coordinate representation. In addition, we compute Riemannian, sectional, Ricci, and scalar curvatures induced by Onsager response operators in hydrodynamics. For this reason, we refer to these curvature quantities as macroscopic curvatures. In future work, we could apply the macroscopic curvatures in studying physical behaviors of hydrodynamics. For example, we will estimate such macroscopic curvatures to study free--energy dissipation toward macroscopic diffusion processes, including super--Brownian motion and stochastic Fokker--Planck dynamics \cite{PHC, Dean, WD}. Meanwhile, it is worth mentioning that the quantities in density manifolds also connects with the ones often studied in information geometry \cite{Amari2009,AyJostLeSchwachhoefer2017,Zhang2004}. The information geometry often studies the geometric structures of free energies, such as Bregman divergences in the $L^2$ space. {We note that} the generalized density manifold investigates the metric operator from an elliptic operator with a mobility function, where the mobility function relates to the Hessian of free energies by the local Einstein relation. Following these connections, one observes that curvatures in density manifolds contain high order derivatives of the Bregman potential. More details connecting information geometry and Riemannian calculuses on density manifolds are left in the future work. 

The paper is organized as follows. Section \ref{sec2} briefly reviews hydrodynamics, Onsager reciprocal relations, and hydrodynamical density manifolds. Section \ref{sec3} derives the Levi--Civita connection, parallel transport equations, and curvature tensors for generalized density manifolds. In section \ref{sec4}, we provide explicit Riemannian and sectional curvatures in one--dimensional settings. Section \ref{sec5} presents several concrete examples of metrics and sectional curvatures in density manifolds.
\section{Onsager reciprocal relations and hydrodynamical density manifolds}\label{sec2}
In this section, we first consider a hydrodynamic description of an out--of--equilibrium physical system. For simplicity of discussion, we focus on a single conservation law of density function. We next review the Onsager reciprocal relation, including the Onsager response operator and the force function \cite{MFT}. It allows us to define an infinite dimensional manifold on density space, namely a hydrodynamical density manifold. We last present the inner product, gradient operator, and arc length of curves defined in the hydrodynamical density manifold \cite{C1, O1}. 
\subsection{Hydrodynamic description}
We consider an example in \cite{MFT}. Throughout this paper, we define $\Omega= \mathbb{R}^d$, a $d$--dimensional Euclidean space. Let $x\in \Omega$  be a spatial variable, and $t>0$ be a time variable. Denote $\int=\int_\Omega$ as the integration symbol over the spatial domain $\Omega$, and write $dx$ as {the differential} in Euclidean space.  

Consider a physical system on $\Omega$ at the macroscopic level, which is fully characterized by the density variable $\rho(t)=\rho(t,\cdot)\in\mathbb{R}_+$ and its local current function $J(\rho)\in \mathbb{R}^d$. Assume that the evolution of the density function satisfies the continuity equation:
\begin{equation*}
\partial_t\rho(t)+\nabla\cdot J(\rho(t))=0, \quad \rho(0)=\rho_0,
\end{equation*}
where we omit the dependence on the $x$ variable for variable $\rho$ and $\rho_0$ is an initial density function with $\int \rho_0 dx=1$, $\rho_0\geq 0$. The continuity equation is in a conservation form, meaning that 
\begin{equation*}
\int \rho(t) dx=\int \rho_0 dx=1. 
\end{equation*}
Consider a diffusion system such that the current satisfies 
\begin{equation*}
J(\rho):=\chi(\rho) E-D(\rho)\nabla\rho, 
\end{equation*}
where $\chi\in C^{\infty}(\mathbb{R}_+; \mathbb{R}^{d\times d})$ is a {symmetric positive} definite mobility matrix, i.e.,  $\chi(\rho)\succ 0$, $D\in C^{\infty}(\mathbb{R}_+; \mathbb{R}^{d\times d})$ is a {symmetric positive} definite diffusion matrix, i.e.,  $D(\rho)\succ 0$, and $E\in C^{\infty}(\Omega; \mathbb{R}^d)$ is an external vector field. Combining the above facts, the time evolution of the continuity equation associated with a diffusion system satisfies 
\begin{equation}\label{hydro}
\partial_t\rho(t)+\nabla\cdot(\chi(\rho(t)) E)=\nabla\cdot(D(\rho(t))\nabla \rho(t)).
\end{equation}
The diffusion coefficient $D$ and transport coefficient $\chi$ are matrices, defined by the local Einstein relation: 
\begin{equation}\label{E1}
D(\rho):=\chi(\rho)f''(\rho), 
\end{equation}
where $f\in C^{2}(\mathbb{R})$ is a convex function. We next characterize the equilibrium state of the system \eqref{hydro}. Denote $\pi\in C^{\infty}(\Omega; \mathbb{R}_+)$ as an equilibrium state, which is the stationary solution of equation \eqref{hydro}. This means 
\begin{equation*}
\nabla\cdot J(\pi)=\nabla\cdot(\chi(\pi)E)-\nabla\cdot(D(\pi)\nabla\pi)=0. 
\end{equation*}
We assume that $\pi$ is a unique stationary solution for equation \eqref{hydro} with $J(\pi)=0$. The external vector field $E$ determines the equilibrium state $\pi$. From now on, we only consider an inhomogeneous equilibrium, where the external vector field is a gradient vector 
\begin{equation}\label{E2}
E(x)=-\nabla U(x), \quad \textrm{$U\in C^{\infty}(\Omega)$ is a potential function}.
\end{equation}
\begin{remark}
{
We emphasize that in this paper we only consider dynamics generated by the gradient vector field $E$. 
This assumption guarantees micro-reversibility of the underlying stochastic process and yields a gradient flow formulation in the density space, which in turn determines the metric structure on the density manifold. In general, if $E$ is not a gradient vector field, the evolution \eqref{hydro} becomes a non-gradient dynamical system containing both gradient and non-gradient components on the density manifold. Such situations have been studied in macroscopic fluctuation theory (MFT) and the GENERIC framework. For simplicity, we restrict attention to the gradient flow--induced density manifold and derive the corresponding geometric calculations.}
\end{remark}
\subsection{Onsager reciprocal relations} 
We next demonstrate that equation \eqref{hydro} with conditions \eqref{E1} and \eqref{E2} satisfies the Onsager reciprocal relations. 

We need some notations. Given a functional $F\in C^{\infty}(\mathcal{P}_+, \mathbb{R})$, we write the $L^2$ first and second variational operators as follows. Denote a smooth testing function $h\in C^{\infty}(\Omega)$ and a scalar $\epsilon\in \mathbb{R}$. Define 
\begin{equation*}
\frac{d}{d\epsilon}F(\rho+\epsilon h)|_{\epsilon=0}:=\int \frac{\delta}{\delta \rho}F(\rho)(x) h(x)dx, 
\end{equation*}
and 
\begin{equation*}
\frac{d^2}{d\epsilon^2}F(\rho+\epsilon h)|_{\epsilon=0}:=\int\int \frac{\delta^2}{\delta \rho^2}F(\rho)(x,y) h(x)h(y)dxdy, 
\end{equation*}
where $\frac{\delta}{\delta\rho}F$ is the $L^2$ first variation operator of $F$, and $\frac{\delta^2}{\delta \rho^2}F$ is the $L^2$ second variation operator of $F$. 

Denote a free energy functional as:
\begin{equation*}
\mathrm{D}_{f}(\rho, \pi)=\int \Big[f(\rho)-f(\pi)-f'(\pi)(\rho-\pi)\Big] dx, 
\end{equation*}
where $f\in C^2(\mathbb{R};\mathbb{R})$ is a strictly convex function. {E.g., $f$ can be an entropy function, widely studied in thermodynamics; see examples in section \ref{sec5}.}
We note that $\mathrm{D}_{f}$ is a Bregman divergence between current density $\rho$ and equilibrium $\pi$ in $L^2$ space. We note that
\begin{equation*}
\frac{\delta}{\delta\rho}\mathrm{D}_{f}(\rho, \pi)=f'(\rho)-f'(\pi), 
\end{equation*}
where $\frac{\delta}{\delta \rho}$ is the $L^2$ first variation operator. In the literature \cite{MFT}, $\nabla\frac{\delta}{\delta\rho}\mathrm{D}_{f}(\rho, \pi)$ is named the thermodynamic force.
\begin{proposition}[Onsager reciprocal relations \cite{MFT}]
Assume that $\chi(\pi)\in\mathbb{R}^{d\times d}$ is a positive definite matrix. Then the current of equation \eqref{hydro} with conditions \eqref{E1} and \eqref{E2} is proportional to the thermodynamic force with a mobility function:
\begin{equation*}
J(\rho)=-\chi(\rho)\nabla\frac{\delta}{\delta\rho}\mathrm{D}_{f}(\rho, \pi). 
\end{equation*}
In other words, equation \eqref{hydro} can be rewritten as follows: 
\begin{equation}\label{gradient}
\partial_t\rho(t)=-\nabla\cdot J(\rho(t))=\nabla\cdot\left(\chi(\rho(t))\nabla \frac{\delta}{\delta\rho}\mathrm{D}_{f}(\rho(t), \pi)\right).
\end{equation}
Formulation \eqref{gradient} is often called the {\em Onsager gradient flow}.
\end{proposition}
\begin{proof}
From condition \eqref{E1}, we have
\begin{equation}\label{E0}
D(\rho)\nabla \rho=\chi(\rho)f''(\rho)\nabla \rho=\chi(\rho)\nabla f'(\rho), 
\end{equation}
where we use the fact that $f''(\rho)\nabla\rho=\nabla f'(\rho)$. In addition, conditions \eqref{E1}, \eqref{E2} imply that 
\begin{equation*}
J(\pi)=-\chi(\pi)\nabla U-\chi(\pi)f''(\pi)\nabla \pi=-\chi(\pi)\nabla\left(U+f'(\pi)\right)=0, 
\end{equation*}
where we use the fact that $f''(\pi)\nabla\pi=\nabla f'(\pi)$. Since the matrix $\chi(\pi)$ is positive definite, we have \begin{equation}\label{E3}
E(x)=-\nabla U(x)=\nabla f'(\pi(x)). 
\end{equation}
From equations \eqref{E0} and \eqref{E3}, we have 
\begin{equation*}
\begin{split}
J(\rho)=&\chi(\rho) E-D(\rho)\nabla\rho=-\chi(\rho) \nabla U-D(\rho)\nabla\rho\\
=&-\chi(\rho)\nabla U-\chi(\rho)\nabla f'(\rho)=-\chi(\rho)\nabla \big(U+f'(\rho)\big)\\
=&-\chi(\rho)\nabla \big(f'(\rho)-f'(\pi)\big). 
\end{split}
\end{equation*}
From the fact that $\frac{\delta}{\delta \rho}\mathrm{D}_f(\rho, \pi)=f'(\rho)-f'(\pi)$, we finish the proof. 
\end{proof}
Based on formulation \eqref{gradient}, the following free energy dissipation property holds. For vectors $u$, $v\in\mathbb{R}^{d}$, we denote $(u ,\chi(\rho)v)=\sum_{i,j=1}^d u_iv_j\chi_{ij}(\rho)$.
\begin{proposition}[{de Bruijn's identity with mobilities}]
Suppose $\rho(t)$ satisfies equation \eqref{hydro} with conditions \eqref{E1} and \eqref{E2}. Then 
\begin{equation*}
\begin{split}
\frac{d}{dt}\mathrm{D}_{f}(\rho(t), \pi)=&-\int \left(\nabla\frac{\delta}{\delta\rho}\mathrm{D}_f(\rho(t), \pi), \chi(\rho(t))\nabla \frac{\delta}{\delta\rho}\mathrm{D}_f(\rho(t), \pi)\right)dx\leq 0. 
\end{split}
\end{equation*}
\end{proposition}
\begin{proof}
We note that 
\begin{equation*}
\begin{split}
\frac{d}{dt}\mathrm{D}_{f}(\rho, \pi)=&\int \frac{\delta}{\delta\rho}\mathrm{D}_f(\rho(t), \pi)\cdot \partial_t\rho(t) dx\\
=&\int \frac{\delta}{\delta\rho}\mathrm{D}_f(\rho(t), \pi) \cdot \nabla\cdot(\chi(\rho(t))\nabla \frac{\delta}{\delta\rho}\mathrm{D}_{f}(\rho(t), \pi)) dx\\
=&-\int \left(\nabla\frac{\delta}{\delta\rho}\mathrm{D}_f(\rho(t), \pi), \chi(\rho(t))\nabla \frac{\delta}{\delta\rho}\mathrm{D}_f(\rho(t), \pi)\right)dx\leq 0,
\end{split}
\end{equation*}
where the third equality is from the integration by parts formula. The last inequality is based on the fact that {the matrix function $\chi(\rho)$} is positive definite. 
\end{proof}
In physics, the dissipation of free energy represents the second law of thermodynamics. The free energy dissipation is based on the gradient flow structure in Onsager reciprocal relation. In geometry, Onsager gradient flow \eqref{gradient} induces an infinite dimensional Riemannian manifold on the density space. In other words, the hydrodynamic evolution \eqref{hydro} is the steep descent direction. Shortly, we present the definition of the Riemannian metric. To do so, we denote the {\em Onsager response operator} as $-\Delta_{\chi}\colon C^{\infty}(\Omega)\rightarrow C^{\infty}(\Omega)$. For any function $\Phi\in C^{\infty}(\Omega)$, define   
\begin{equation*}
-\Delta_{\chi}\Phi:=-\nabla\cdot(\chi(\rho)\nabla \Phi).  
\end{equation*}
\subsection{Hydrodynamical density manifold}
We next introduce the infinite dimensional density manifold or metric space based on Onsager reciprocal relations \eqref{gradient}. In the metric space, we define gradient operators, arc lengths of curves, and distance functionals between two density functions. 

Denote the smooth positive density space as 
\begin{equation*}
\sP_+:=\big\{\rho \in C^{\infty}(\Omega) \colon \int \rho dx=1, \quad \rho> 0\big\}.
\end{equation*}
Define the smooth tangent space at density function $\rho\in\mathcal{P}_+$ as
\begin{equation*}
T_\rho\mathcal{P}_+ = \big\{\sigma \in C^{\infty}(\Omega)\colon  \int \sigma dx=0 \big\}.
\end{equation*}
We note that the operator $(-\Delta_{\chi})$ is a symmetric positive definite. In other words, for two functions $\Phi_1$, $\Phi_2\in C^{\infty}(\Omega)$, then  
\begin{equation*}
\int \Phi_1\cdot (-\Delta_{\chi}\Phi_2)dx=\int (\nabla\Phi_1, \chi(\rho)\nabla\Phi_2)dx=\int \Phi_2\cdot (-\Delta_{\chi}\Phi_1)dx.
\end{equation*}

Denote the pseudo-inverse of the Onsager response operator as 
\begin{equation*}
-\Delta_{\chi}^{\dagger}\colon C^{\infty}(\Omega)\rightarrow C^{\infty}(\Omega),\end{equation*}
where $\dagger$ represents the pseudo-inverse operator. Given a function $\sigma\in T_\rho \sP_+$, and a function $\Phi\in C^{\infty}(\Omega)$, we write
\begin{equation*}
-\nabla\cdot(\chi(\rho)\nabla \Phi)=-\Delta_{\chi}\Phi=\sigma, \qquad \Phi:=-\Delta_\chi^{\dagger}\sigma. 
\end{equation*}
For any function $\Phi\in C^{\infty}(\Omega)$ up to a constant shift, we denote a tangent vector field in $T_\rho\mathcal{P}_+$ as
\begin{equation*}
\V_{\Phi}:=-\Delta_{\chi}\Phi=-\nabla\cdot(\chi(\rho) \nabla \Phi)\in T_\rho\mathcal{P}_+. 
\end{equation*}
Assume that the map $\Phi\rightarrow \V_{\Phi}$ is an isomorphism $C^{\infty}(\Omega)/\mathbb{R}\rightarrow T_{\rho}\mathcal{P}_+$. 

\begin{definition}[Hydrodynamical metric tensor]
Define the inner product $\g:\mathcal{P}_+\times T_\rho\mathcal{P}_+\times T_\rho\mathcal{P}_+\rightarrow\mathbb{R}$ as 
\begin{equation*}
\begin{split}
\g(\rho)(\V_{\Phi_1}, \V_{\Phi_2}):=\langle \V_{\Phi_1}, \V_{\Phi_2}\rangle(\rho):=-\int \V_{\Phi_1}\cdot\Delta_\chi^{\dagger}\V_{\Phi_2} dx,
\end{split}
\end{equation*}
where $\Phi_k\in C^\infty(\Omega)$, $k=1,2$, and 
\begin{equation*}
 \V_{\Phi_k}=-\Delta_{\chi}\Phi_k=-\nabla\cdot(\chi(\rho) \nabla \Phi_k)\in T_\rho\mathcal{P}_+. 
\end{equation*}
In other words,   
\begin{equation*}
\begin{split}
\langle \V_{\Phi_1}, \V_{\Phi_2}\rangle(\rho)=&-\int \Delta_\chi\Phi_1 \cdot \Delta_{\chi}^{\dd}( \Delta_{\chi}(\Phi_2))dx\\
=&-\int \Phi_1\cdot \Delta_\chi(\Delta_{\chi}^{\dd}(\Delta_{\chi}(\Phi_2)))dx\\
=&-\int \Phi_1\cdot \Delta_{\chi}\Phi_2dx\\ 
=& \int (\nabla \Phi_1, \chi(\rho)\nabla\Phi_2)dx,
\end{split}
\end{equation*}
where we use the fact that $\Delta_\chi(\Delta_{\chi}^{\dd}(\Delta_{\chi}\Phi))=\Delta_{\chi}\Phi$, for any {function} $\Phi\in C^\infty(\Omega)$, and the integration by parts formula. 
\end{definition}
The inner product $\g$ introduces an infinite dimensional Riemannian metric on the density space $\mathcal{P}_+$; see related studies in \cite{C1}. We note that the metric $\g$ is derived from Onsager reciprocal relations and hydrodynamics. In this reason, we call $(\mathcal{P}_+, \g)$ the {\em hydrodynamical density manifold}. 

We next present some quantities in density manifold $(\mathcal{P}_+, \g)$. 
We first study the gradient operator in the density manifold. Denote $\bar\grad\colon \sP_+\times C^\infty(\mathcal{P}_+; \mathbb{R})\rightarrow T_\rho\sP_+$. 
\begin{proposition}[Gradient operators]
Denote an energy functional $F\in C^\infty (\mathcal{P}_+; \mathbb{R})$. The gradient operator of functional $F$ in $(\mathcal{P}_+, \g)$ satisfies  
\begin{equation*}
\bar\grad F(\rho)=-\Delta_{\chi}\frac{\delta}{\delta \rho} F(\rho)=-\nabla\cdot(\chi(\rho) \nabla\frac{\delta}{\delta \rho} F(\rho)).\end{equation*}
In particular, if $F(\rho)=\mathrm{D}_{f}(\rho, \pi)$, then   
\begin{equation}\label{Onsager}
\begin{split}
\bar\grad \mathrm{D}_{f}(\rho, \pi)=&-\Delta_{\chi}\frac{\delta}{\delta\rho}\mathrm{D}_{f}(\rho, \pi)
=-\nabla\cdot J(\rho).
\end{split}
\end{equation}
\end{proposition}
\begin{proof}
The proof follows from the definition of the gradient operator in a Riemannian manifold. Since $(-\Delta_{\chi})^{\dd}$ is a Riemannian metric tensor in $(\mathcal{P}_+, \g)$, then for any function $\Phi\in C^{\infty}(\Omega)/\mathbb{R}$ with $\sigma=\V_\Phi=-\Delta_{\chi}\Phi$, we have 
\begin{equation}\label{Derive_a}
\g(\bar\grad F(\rho), \V_\Phi)=\int \Phi\cdot \bar\grad F(\rho) dx =\int \frac{\delta}{\delta\rho}\mathcal{F}(\rho)\cdot \V_\Phi dx. 
\end{equation}
We note that from the integration by parts, 
\begin{equation*}
\int \frac{\delta}{\delta\rho}\mathcal{F}(\rho)\cdot \V_\Phi dx=\int \frac{\delta}{\delta\rho}\mathcal{F}(\rho)\cdot (-\Delta_{\chi}\Phi) dx
=\int -\Delta_{\chi}\frac{\delta}{\delta\rho}\mathcal{F}(\rho)\cdot \Phi dx.
\end{equation*}
From equation {\eqref{Derive_a}}, we have 
\begin{equation*}
\int \Phi\cdot\left(\bar \grad F(\rho)+\Delta_{\chi}\frac{\delta}{\delta\rho}\mathcal{F}(\rho)\right) dx=0, 
\end{equation*}
for any {function} $\Phi\in C^{\infty}(\Omega)/\mathbb{R}$. Let $\Phi=\bar \grad F(\rho)+\Delta_{\chi}\frac{\delta}{\delta\rho}\mathcal{F}(\rho)$. We have 
\begin{equation*}
\int\left(\bar \grad F(\rho)+\Delta_{\chi}\frac{\delta}{\delta\rho}\mathcal{F}(\rho)\right)^2 dx=0.
\end{equation*}
Thus, 
 \begin{equation*}
\bar \grad F(\rho)=-\Delta_{\chi}\frac{\delta}{\delta \rho}F(\rho),\quad a.e.. 
\end{equation*}
If $F(\rho)=\mathrm{D}_{f}(p, \pi)$, using the formula of $\bar\grad F$ and equation \eqref{gradient}, we show equation \eqref{Onsager}.
\end{proof}
Again, equation \eqref{Onsager} explains the geometric representation in Onsager reciprocal relations. The R.H.S. of hydrodynamics is the steep descent direction of the free energy in $(\mathcal{P}_+, \g)$. In fact, one can define other geometric quantities that are useful in studying the dynamical behaviors of hydrodynamics. For example, one {may informally} define the arc length functional in $(\mathcal{P}_+, \g)$.  
\begin{definition}[Arc length functional]
For any curve $\gamma\in C^{1}([0, T];\sP_+)$, with $T>0$, the arc length $\bar L(\gamma):=L_{\g}(\gamma)$ of curve $\gamma$ in $(\sP_+, \g)$ is defined as   
\begin{equation}\label{L}
\bar L(\gamma):=\int_0^{T} \big|\int \partial_t\gamma(t)\cdot \left(-\Delta^{\dagger}_{\chi(\gamma(t))}\partial_t\gamma(t)\right) dx\big|^{\frac{1}{2}}dt. 
\end{equation}
In other words, let function $\Phi(t)\in C^{\infty}(\Omega)$, $t\in [0,T]$. Consider
\begin{equation*}
\partial_t\gamma(t)+\nabla\cdot(\chi(\gamma(t))\nabla\Phi(t))=0. 
\end{equation*}
Then
\begin{equation*}
\bar L(\gamma)=\int_0^{T} |\int (\nabla\Phi(t), \chi(\gamma(t))\nabla\Phi(t))dx|^{\frac{1}{2}}dt.
\end{equation*}
\end{definition}
One also {informally} formulates the minimal arc length problem between two density functions. The minimal value defines a Riemannian distance. Denote $\bar\Dist=\mathrm{Dist}_\g\colon \sP_+\times\sP_+\rightarrow\mathbb{R}_+$. 
\begin{definition}[Minimal arc length problems]
Given two points $\rho^0$, $\rho^1\in\mathcal{P}_+$. The minimal arc length problem in $(\sP_+, \g)$ satisfies the following optimization problem: 
\begin{equation*}
\bar{\Dist}(\rho^0,  \rho^1):=\inf_{\gamma\in C^{\infty}([0,1];\mathcal{P}_+)} \Big\{\int_0^1 \big|\int \partial_t\gamma(t)\cdot \Big(-\Delta^{\dagger}_{\chi(\gamma(t))}\partial_t\gamma(t)\Big) dx\big|^{\frac{1}{2}}dt\colon \gamma(0)=\rho^0,~~\gamma(1)=\rho^1\Big\},
\end{equation*}
where $\bar L$ is the arc length function defined in \eqref{L} and the minimal is over all smooth curves $\gamma(t)\in\mathcal{P}_+$, $t\in [0,1]$, connecting initial and terminal time densities $\gamma(0)=\rho^0$, $\gamma(1)=\rho^1$. 
\end{definition}
The other equivalent of minimal arc length is written as the minimization of the squared norm in $(\mathcal{P}_+, \g)$, known as the least action problem.  
\begin{equation*}
\begin{split}
&\bar{\Dist}(\rho^0,  \rho^1)^2\\=&\inf_{\gamma\in C^{\infty}([0,1];\mathcal{P}_+)} \Big\{\int_0^1 \int \partial_t\gamma(t)\cdot \Big(-\Delta^{\dagger}_{\chi(\gamma(t))}\partial_t\gamma(t)\Big) dxdt\colon \gamma(0)=\rho^0,~~\gamma(1)=\rho^1\Big\}\\
=&\inf_{\gamma\in C^{\infty}([0,1];\mathcal{P}_+)} \Big\{\int_0^1 \int \Big(\nabla\Phi(t), \chi(\gamma(t))\nabla\Phi(t)\Big) dxdt\colon\\
&\hspace{2.5cm} \partial_t\gamma(t)+\nabla\cdot\big(\chi(\gamma(t))\nabla\Phi(t)\big)=0, ~~\gamma(0)=\rho^0,~~\gamma(1)=\rho^1\Big\}. 
\end{split}
\end{equation*}
In the above least action problem, we solve the function $\Phi(t)$ from the continuity equation $\partial_t\gamma(t)=-\Delta_{\chi(\gamma(t))}\Phi(t)$, and then take the infimum over all paths $\gamma(t)$, $t\in [0,1]$ connecting densities $\rho^0$, $\rho^1$. We note that the minimal arc length value $\bar\Dist$ represents the Wasserstein-2 type distance on the density space; see \cite{C1,LiG2}. In particular, if $\chi(\rho)=\rho\mathbb{I}$, the distance functional $\bar\Dist$ defined the classical Wasserstein-2 distance \cite{am2006,vil2008}. If $\chi$ is not a linear function of $\rho$, the distance functional will introduce a class of Wasserstein-2 type distances. 
\begin{remark}{We refer that the existence and well posedness of geodesics in Wasserstein-2 type space has been studied in \cite{C1, O1}. 
E.g., assume the nonlinear mobility matrix function $\chi(\rho) = m(\rho)\mathbb{I}$, $m(\rho)\geq 0$ is a smooth function.
The existence of geodesics between two densities depends on the concavity of $m(\rho)$, i.e., $m''(\rho)\leq 0$. 
In this paper, we present an informal derivation of a geodesic equation, in which we assume that the initial value geodesic exists for a short time.}
\end{remark}
\section{Geometric calculations in hydrodynamical density manifolds}\label{sec3}
In this section, we are ready to present the main result of this paper. We derive Levi--Civita connections, parallel transport, geodesic equations, and curvature tensors in the hydrodynamical density manifold $(\mathcal{P}_+, \g)$. This computation extends the ones in the Wasserstein-2 density manifold \cite{Lafferty, Lott}.
\subsection{Levi--Civita connection, parallel transport and geodesics} 
For simplicity of discussion, we {use} the Eulerian coordinates in fluid mechanics to represent geometric quantities in density manifolds. {The following definitions are needed, which are mobility induced {\em carr{\'e} du champ} or Gamma one operators.} 
\begin{definition}[Gamma one operators {with mobilities}]
 Denote $\rho\in \mathcal{P}_+$, and $\Phi_1$, $\Phi_2\in C^{\infty}(\Omega)$. Define $\Gamma_{\chi}\colon \sP_+\times C^{\infty}(\Omega)\times C^{\infty}(\Omega)\rightarrow C^{\infty}(\Omega)$, such that  
\begin{equation*}
\Gamma_{\chi}(\Phi_1, \Phi_2):= (\nabla\Phi_1, \chi(\rho)\nabla\Phi_2)=\sum_{i,j=1}^d\frac{\partial}{\partial x_i}\Phi_1\frac{\partial}{\partial x_j}\Phi_2\chi_{ij}(\rho). 
\end{equation*}
Define $\Gamma_{\chi'}\colon \sP_+\times C^{\infty}(\Omega)\times C^{\infty}(\Omega)\rightarrow C^{\infty}(\Omega)$, such that 
\begin{equation*}
\Gamma_{\chi'}(\Phi_1, \Phi_2):= (\nabla\Phi_1, \chi'(\rho)\nabla\Phi_2)=\sum_{i,j=1}^d\frac{\partial}{\partial x_i}\Phi_1\frac{\partial}{\partial x_j}\Phi_2\chi'_{ij}(\rho), 
\end{equation*}
where $\chi_{ij}'(\rho)=\frac{\partial}{\partial \rho}\chi_{ij}(\rho)$. Define $\Gamma_{\chi''}\colon \sP_+\times C^{\infty}(\Omega)\times C^{\infty}(\Omega)\rightarrow C^{\infty}(\Omega)$, such that  
\begin{equation*}
\Gamma_{\chi''}(\Phi_1, \Phi_2):= (\nabla\Phi_1, \chi''(\rho)\nabla\Phi_2)=\sum_{i,j=1}^d\frac{\partial}{\partial x_i}\Phi_1\frac{\partial}{\partial x_j}\Phi_2\chi''_{ij}(\rho),
\end{equation*}
where $\chi_{ij}''(\rho)=\frac{\partial^2}{\partial \rho^2}\chi_{ij}(\rho)$. 
\end{definition}
\begin{remark}
In particular, let $\chi=\mathbb{I}$, where $\mathbb{I}\in\mathbb{R}^{d\times d}$ is an identity matrix. Denote  
\begin{equation*}
\Gamma_{1}(\Phi_1, \Phi_2):=\Gamma_{\chi}(\Phi_1,\Phi_2)= (\nabla\Phi_1, \nabla\Phi_2)=\sum_{i=1}^d\frac{\partial}{\partial x_i}\Phi_1\frac{\partial}{\partial x_i}\Phi_2.  
\end{equation*}
Here, the $\Gamma_1$ operator is the {{\em carr{\'e} du champ} or} Gamma one operator in Euclidean space, which was firstly studied in \cite{BE}. Our geometric calculations are built on generalized Gamma one operators; see previous work in \cite{LiG1}.  
\end{remark}

\begin{definition}[Directional derivatives]\label{DR}
Given a function $\Phi\in C^{\infty}(\Omega)$, denote a vector field $\V_{\Phi}=-\Delta_{\chi}\Phi\in {T}_ \rho\mathcal{P}_+$. Denote an energy functional $F\in C^{\infty}(\mathcal{P}_+; \mathbb{R})$. Write the direction derivative of $F$ at direction $\V_\Phi$ as 
\begin{equation*}
\begin{split}
(\V_{\Phi}F)(\rho):=&\frac{d}{d\epsilon}|_{\epsilon=0}F(\rho-\epsilon \Delta_{\chi}\Phi)=-\int \frac{\delta}{\delta \rho} F(\rho) \nabla\cdot(\chi(\rho)\nabla\Phi) dx\\
=&\int (\nabla\frac{\delta}{\delta \rho} F(\rho), \chi(\rho)\nabla\Phi)dx\\
=&\int \Gamma_{\chi}(\frac{\delta}{\delta \rho} F(\rho), \Phi)dx. 
\end{split}
\end{equation*}
We also denote the first order directional derivative of mobility matrix function $\chi$ at direction $\V_\Phi$ as $\V_\Phi\chi=\V_\Phi\chi(\rho)=((\V_\Phi\chi(\rho))_{ij})_{1\leq i,j\leq d}\in C^{\infty}(\Omega; \mathbb{R}^{d\times d})$, such that
\begin{equation*}
(\V_\Phi\chi(\rho))_{ij}:=-\nabla\cdot(\chi(\rho)\nabla\Phi)\chi'_{ij}(\rho).
\end{equation*}
\end{definition}

We first compute commutators of two vector fields in $(\mathcal{P}_+, \g)$. Denote the commutator $[\cdot, \cdot]\colon T_\rho\sP_+\times T_\rho\sP_+\rightarrow T_\rho\sP_+$. {For an energy functional $F\in C^{\infty}(\mathcal{P}_+; \mathbb{R})$, we define 
\begin{equation*}
([\V_{\Phi_1}, \V_{\Phi_2}])F(\rho):=(\V_{\Phi_1}(\V_{\Phi_2}F))(\rho)-(\V_{\Phi_2}(\V_{\Phi_1}F))(\rho).
\end{equation*}}
\begin{lemma}\label{communtator}
Given functions $\Phi_1$, $\Phi_2\in C^{\infty}(\Omega)$ and a functional $F\in C^{\infty}(\mathcal{P}_+, \mathbb{R})$, the commutator $[\V_{\Phi_1}, \V_{\Phi_2}]$ in $(\sP_+, \g)$ satisfies 
\begin{equation}\label{comm}
[\V_{\Phi_1}, \V_{\Phi_2}]F(\rho)= \int \Gamma_{\chi}(\Gamma_{\chi'}(\frac{\delta}{\delta\rho}F(\rho), \Phi_2), \Phi_1)- \Gamma_{\chi}(\Gamma_{\chi'}(\frac{\delta}{\delta\rho}F(\rho), \Phi_1), \Phi_2) dx.
\end{equation}
Equivalently, 
\begin{equation}\label{comm_ex}
[\V_{\Phi_1}, \V_{\Phi_2}]=-\Delta_{\V_{\Phi_1}\chi}\Phi_2+\Delta_{\V_{\Phi_2}\chi}\Phi_1, 
\end{equation}
where we denote 
\begin{equation*}
\Delta_{\V_{\Phi_1}\chi}\Phi_2:=-\nabla\cdot\Big( \nabla\cdot(\chi(\rho)\nabla\Phi_1)\chi'(\rho)\nabla\Phi_2\Big).
\end{equation*}
\end{lemma}
\begin{proof}
We note that 
\begin{equation}\label{lem1_int}
\begin{split}
([\V_{\Phi_1}, \V_{\Phi_2}])F(\rho)=&(\V_{\Phi_1}(\V_{\Phi_2}F))(\rho)-(\V_{\Phi_2}(\V_{\Phi_1}F))(\rho)\\
=&\frac{d}{d\epsilon}|_{\epsilon=0} (\V_{\Phi_2}F)( \rho-\epsilon\Delta_{\chi}\Phi_1)-\frac{d}{d\epsilon}|_{\epsilon=0} (\V_{\Phi_1}F)(\rho-\epsilon_2 \Delta_{\chi}\Phi_2)  \\
=&\int \Gamma_{\chi}(\frac{\delta}{\delta \rho} \V_{\Phi_2}F(\rho), \Phi_1)dx-\int \Gamma_{\chi}(\frac{\delta}{\delta \rho} \V_{\Phi_1}F(\rho), \Phi_2)dx. 
\end{split}
\end{equation}
From Definition \ref{DR}, we have
\begin{equation*}
\begin{split}
\frac{\delta}{\delta \rho} \V_{\Phi_2}F(\rho)(x)=&-\frac{\delta}{\delta \rho}\int \frac{\delta}{\delta \rho} F(\rho)(x)\nabla\cdot(\chi(\rho(x))\nabla\Phi_2(x)) dx\\
=& -\int \frac{\delta^2}{\delta \rho^2} F(\rho)(x,y) \nabla\cdot(\chi(\rho(y))\nabla\Phi_2(y)) dy+ (\nabla\frac{\delta}{\delta \rho} F(\rho)(x), \chi'(\rho(x))\nabla\Phi_2(x))\\
=& -\int \frac{\delta^2}{\delta \rho^2} F(\rho)(x,y) \Delta_{\chi}\Phi_2(y) dy+ \Gamma_{\chi'}(\frac{\delta}{\delta \rho} F(\rho), \Phi_2)(x).
\end{split}
\end{equation*}
Thus
\begin{equation*}
\begin{split}
\int \Gamma_{\chi}(\frac{\delta}{\delta \rho} \V_{\Phi_2}F(\rho), \Phi_1)dx=&\int \int  \frac{\delta^2}{\delta \rho^2} F(\rho)(x,y)\cdot\Delta_{\chi}\Phi_1(x) \Delta_{\chi}\Phi_2(y) dxdy\\
&+ \int \Gamma_{\chi}(\Gamma_{\chi'}(\frac{\delta}{\delta \rho} F(\rho), \Phi_2), \Phi_1) dx, 
\end{split}
\end{equation*}
where $\frac{\delta^2}{\delta\rho^2}F(\rho)(x,y)$ is the $L^2$ second variation operator of function $F(\rho)$ at a point $(x,y)\in \Omega \times \Omega$. Switching index $1$ and $2$, we can compute the term $\int \Gamma_{\chi}(\frac{\delta}{\delta \rho} \V_{\Phi_1}F(\rho), \Phi_2)dx$. Combing above terms into equation \eqref{lem1_int}, we have
\begin{equation*}
\begin{split}
&([\V_{\Phi_1}, \V_{\Phi_2}])F(\rho)\\=&\quad\int \int  \frac{\delta^2}{\delta \rho^2} F(\rho)(x,y)\cdot\Delta_{\chi}\Phi_1(x) \Delta_{\chi}\Phi_2(y) dxdy+ \int \Gamma_{\chi}(\Gamma_{\chi'}(\frac{\delta}{\delta \rho} F(\rho), \Phi_2), \Phi_1) dx\\
&-\int \int  \frac{\delta^2}{\delta \rho^2} F(\rho)(x,y)\cdot\Delta_{\chi}\Phi_1(y) \Delta_{\chi}\Phi_2(x) dxdy- \int \Gamma_{\chi}(\Gamma_{\chi'}(\frac{\delta}{\delta \rho} F(\rho), \Phi_1), \Phi_2) dx\\
=& \int \Gamma_{\chi}(\Gamma_{\chi'}(\frac{\delta}{\delta \rho} F(\rho), \Phi_2), \Phi_1) dx- \int \Gamma_{\chi}(\Gamma_{\chi'}(\frac{\delta}{\delta \rho} F(\rho), \Phi_1), \Phi_2) dx, 
\end{split}
\end{equation*}
where the last equality uses the fact that the second variational operator is symmetric, i.e. $\frac{\delta^2}{\delta\rho^2}F(\rho)(x,y)=\frac{\delta^2}{\delta\rho^2}F(\rho)(y,x)$, for any $x,y\in \Omega$. We finish the proof by using integration by parts twice in the above formula.  
\end{proof}

We next compute the Levi--Civita connection in $(\mathcal{P}_+, \g)$. Denote the Levi--Civita connection as $\bar \nabla\colon \sP_+\times T_\rho\sP_+\times T_\rho\sP_+\rightarrow T_\rho\sP_+$. 
\begin{lemma}
Given functions $\Phi_1$, $\Phi_2\in C^{\infty}(\Omega)$, the Levi--Civita connection $\bar \nabla$ in $(\sP_+, \g)$ satisfies
\begin{equation}\label{V12}
\bar\nabla_{\V_{\Phi_1}}\V_{\Phi_2}:=-\frac{1}{2}\Big\{\Delta_{{{\V_{\Phi_1}}\chi}}\Phi_2-\Delta_{{{\V_{\Phi_2}}\chi}}\Phi_1+\Delta_{\chi}\Gamma_{\chi'}(\Phi_1,\Phi_2)\Big\}.
\end{equation}
\end{lemma}
\begin{proof}
We note that $\langle \V_{\Phi_2}, \V_{\Phi_3}\rangle(\rho)=-\int \Phi_2\cdot\Delta_{\chi}\Phi_3 dx$. Thus 
\begin{equation}\label{V}
\begin{split}
\V_{\Phi_1}\langle \V_{\Phi_2}, \V_{\Phi_3}\rangle(\rho)
=& -\int \Phi_2\cdot\Delta_{\V_{\Phi_1}\chi}\Phi_3 dx.
\end{split}
\end{equation}
From the definition of the Levi--Civita connection by the Koszul formula \cite{Lott}, we have
\begin{equation}\label{V11}
\begin{split}
2\langle\bar\nabla_{\V_{\Phi_1}}\V_{\Phi_2}, \V_{\Phi_3}\rangle=&\quad \V_{\Phi_1}\langle \V_{\Phi_2}, \V_{\Phi_3}\rangle +\V_{\Phi_2}\langle \V_{\Phi_3}, \V_{\Phi_1}\rangle -\V_{\Phi_3}\langle \V_{\Phi_1}, \V_{\Phi_2}\rangle \\
&+\langle \V_{\Phi_3}, [\V_{\Phi_1}, \V_{\Phi_2}]\rangle-  \langle \V_{\Phi_2}, [\V_{\Phi_1}, \V_{\Phi_3}]\rangle-\langle \V_{\Phi_1}, [\V_{\Phi_2}, \V_{\Phi_3}]\rangle.
\end{split}
\end{equation}
Substituting equations \eqref{comm_ex}, \eqref{V} into \eqref{V11}, we obtain 
\begin{equation}\label{V1}
\begin{split}
&-2\langle\bar\nabla_{\V_{\Phi_1}}\V_{\Phi_2}, \V_{\Phi_3}\rangle\\=&\quad\int \Phi_2\cdot\Delta_{\V_{\Phi_1}\chi}\Phi_3dx+\int \Phi_1\cdot \Delta_{\V_{\Phi_2}\chi}\Phi_3dx -\int\Phi_1\cdot\Delta_{\V_{\Phi_3}\chi}\Phi_2 dx \\
&+\int \Phi_3\cdot\Delta_{\V_{\Phi_1}\chi}\Phi_2 dx-\int \Phi_3\cdot\Delta_{\V_{\Phi_2}\chi}\Phi_1 dx \\ 
&-\int \Phi_2\cdot\Delta_{\V_{\Phi_1}\chi}\Phi_3 dx+\int \Phi_2\cdot\Delta_{\V_{\Phi_3}\chi}\Phi_1 dx-\int \Phi_1\cdot\Delta_{\V_{\Phi_2}\chi}\Phi_3 dx+\int \Phi_1\cdot\Delta_{\V_{\Phi_3}\chi}\Phi_2 dx.
\end{split}
\end{equation}
 We use the fact that $\int \Phi_i \cdot\Delta_{\V_{\Phi_k}\chi} \Phi_j dx= \int \Phi_j\cdot\Delta_{\V_{\Phi_k\chi}} \Phi_i dx$ from the intergation by parts, for $i,j,k\in\{1,2,3\}$. By {cancelling} the equivalent terms, {such as
$\int \Phi_2\cdot\Delta_{\V_{\Phi_1}\chi}\Phi_3dx$, $\int \Phi_1\cdot \Delta_{\V_{\Phi_2}\chi}\Phi_3dx$, $\int\Phi_1\cdot\Delta_{\V_{\Phi_3}\chi}\Phi_2 dx$}, in \eqref{V1}, we derive
\begin{equation}\label{a}
\begin{split}
&\langle \bar\nabla_{\V_{\Phi_1}}\V_{\Phi_2}, \V_{\Phi_3}\rangle\\
=&-\frac{1}{2}\Big(\int \Phi_3\cdot\Delta_{\V_{\Phi_1}\chi}\Phi_2 dx- \int \Phi_3\cdot\Delta_{\V_{\Phi_2}\chi}\Phi_1 dx\Big)-\frac{1}{2}\int \Phi_1\cdot\Delta_{\V_{\Phi_3}\chi}\Phi_2 dx.
\end{split}
\end{equation}
Clearly, one can {rewrite} equation \eqref{a} as follows. Using the integration by parts twice, we have 
\begin{equation}\label{claim}
\begin{split}
\int \Phi_1\cdot\Delta_{\V_{\Phi_3}\chi}\Phi_2 dx
=& \int (\nabla\Phi_1, \chi'\nabla\Phi_2)\cdot\Delta_\chi\Phi_3 dx\\
=&\int \Phi_3\cdot\Delta_{\chi}\Gamma_{\chi'}(\Phi_1, \Phi_2) dx.
\end{split}
\end{equation}
Substituting equation \eqref{claim} into \eqref{a}, we derive equation \eqref{V12}. 
The formula \eqref{claim} can be verified as below:
\begin{equation*}
\begin{split}
\int \Phi_1\cdot\Delta_{\V_{\Phi_3}\chi}\Phi_2 dx
=&\int \Phi_1 \nabla\cdot(\V_{\Phi_3}\chi \nabla\Phi_2)dx\\
=&-\int (\nabla\Phi_1,  \V_{\Phi_3}\chi\nabla\Phi_2) dx\\
=&\int (\nabla\Phi_1,  \chi'(\rho)\nabla\Phi_2) \nabla\cdot(\chi(\rho) \nabla\Phi_3) dx\\
=&\int \Phi_3 \nabla\cdot\Big(\chi(\rho)\nabla ((\nabla\Phi_1, \chi'(\rho)\nabla\Phi_2) )  \Big)dx\\
=&\int\Phi_3\cdot\Delta_{\chi}\Gamma_{\chi'}(\Phi_1,\Phi_2)  dx. 
\end{split}
\end{equation*}
Combining equations \eqref{a} and \eqref{claim}, we obtain
\begin{equation}\label{newclaim}
\begin{split}
&\langle \bar\nabla_{\V_{\Phi_1}}\V_{\Phi_2}, \V_{\Phi_3}\rangle\\
=&-\frac{1}{2}\Big(\int \Phi_3\cdot \Delta_{\V_{\Phi_1}\chi}\Phi_2 dx- \int \Phi_3\cdot \Delta_{\V_{\Phi_2}\chi}\Phi_1 dx\Big)-\frac{1}{2}\int \Phi_3\cdot\Delta_{\chi}\Gamma_{\chi'}(\Phi_1,\Phi_2) dx\\
=&\frac{1}{2}\int \Phi_3\cdot\Big\{-\Delta_{\V_{\Phi_1}\chi}\Phi_2+ \Delta_{\V_{\Phi_2}\chi}\Phi_1-\Delta_{\chi}\Gamma_{\chi'}(\Phi_1,\Phi_2)\Big\}dx. 
\end{split}
\end{equation}
This finishes the proof. 
 \end{proof}
\begin{lemma}\label{lemma5}
The following equality holds: 
\begin{equation*}
\bar\nabla_{\V_{\Phi_1}}\V_{\Phi_2}+\bar\nabla_{\V_{\Phi_2}}\V_{\Phi_1}=\V_{\Gamma_{\chi'}(\Phi_1, \Phi_2)}.
\end{equation*}
\end{lemma}
\begin{proof}
Since \eqref{V12} holds and $$\bar\nabla_{\V_{\Phi_1}}\V_{\Phi_2}+ \bar\nabla_{\V_{\Phi_2}}\V_{\Phi_1}=-\Delta_{\chi}\Gamma_{\chi'}(\Phi_1, \Phi_2)=\V_{\Gamma_{\chi'}(\Phi_1, \Phi_2)},$$
then we prove the result.
\end{proof}
\begin{lemma}\label{lemma4}
The Levi--Civita connection in $(\sP_+, \g)$ is given as below. For functions $\Phi_1$, $\Phi_2$, $\Phi_3\in C^{\infty}(\Omega)$,   
\begin{equation}\label{LC}
\begin{split}
&\langle \bar\nabla_{\V_{\Phi_1}}\V_{\Phi_2}, \V_{\Phi_3}\rangle\\
=&\frac{1}{2}\int\Big\{\Gamma_{\chi}(\Gamma_{\chi'}(\Phi_2, \Phi_3), \Phi_1)-\Gamma_{\chi}(\Gamma_{\chi'}(\Phi_1, \Phi_3), \Phi_2)
+\Gamma_{\chi}(\Gamma_{\chi'}(\Phi_1, \Phi_2), \Phi_3)\Big\}dx.
\end{split}
\end{equation}
\end{lemma}
\begin{proof}
The Levi--Civita connection \eqref{LC} is derived from equation \eqref{newclaim}. Firstly, we compute 
\begin{equation*}
\begin{split}
-\int \Phi_3\cdot\Delta_{\V_{\Phi_1}\chi}\Phi_2dx=&\int \Phi_3\nabla\cdot\Big(\chi'(\rho)\nabla\cdot(\chi(\rho)\nabla\Phi_1)\nabla\Phi_2\Big)dx\\
=&-\int (\nabla\Phi_2, \chi'(\rho)\nabla \Phi_3)\nabla\cdot(\chi(\rho)\nabla\Phi_1)dx\\
=&\int \Big(\nabla[(\nabla\Phi_2, \chi'(\rho)\nabla \Phi_3)], \chi(\rho)\nabla \Phi_1\Big)dx\\
=&\int \Gamma_{\chi}(\Gamma_{\chi'}(\Phi_2, \Phi_3), \Phi_1) dx. 
\end{split}
\end{equation*}
Secondly, we switch indices $1$ and $2$ in the above formula to obtain  
\begin{equation*}
\begin{split}
-\int \Phi_3\cdot\Delta_{\V_{\Phi_2}\chi}\Phi_1dx=&\int \Gamma_{\chi}(\Gamma_{\chi'}(\Phi_1, \Phi_3), \Phi_2) dx.  
\end{split}
\end{equation*}
Lastly, we have 
\begin{equation*}
-\int \Phi_3\cdot\Delta_{\chi}\Gamma_{\chi'}(\Phi_1,\Phi_2)dx=\int \Gamma_{\chi}(\Gamma_{\chi'}(\Phi_1, \Phi_2), \Phi_3) dx. 
\end{equation*}
Combining the above derivations, we finish the proof. 
\end{proof}

We next compute the parallel transport in density manifolds. Let $ \gamma\colon [0,T]\rightarrow \mathcal{P}_+$ be a smooth curve, with a parameter $T>0$. 
Denote $\V_\Phi$ as the tangent direction of the curve $\gamma(t)$ at time $t$. I.e., $\frac{d \gamma(t)}{dt}=-\Delta_{\chi(\gamma(t))}\Phi(t)=\V_{\Phi(t)}$. 
Consider a vector field $\V_{\eta}$ given by $\eta(t)\in C^{\infty}(\Omega)$. Then the equation for $\V_{\eta}$ to be parallel along $\gamma(t)$ satisfies   
\begin{equation*}
\V_{\partial_t\eta}+ \bar\nabla_{\V_{\Phi}}\V_\eta=0.
\end{equation*}

\begin{theorem}[Parallel transport equations]
For $\V_\eta$ to be parallel along the curve $\gamma$, then the following system of parallel transport equations holds: 
\begin{equation}\label{parallel}
\partial_t\gamma+\Delta_{\chi}\Phi=0,\quad \Delta_{\chi}\partial_t\eta+\frac{1}{2}\Big\{\Delta_{{{\V_{\Phi}}\chi}}\eta-\Delta_{{{\V_{\eta}}\chi}}\Phi+\Delta_{\chi}\Gamma_{\chi'}(\Phi,\eta)\Big\}=0.
\end{equation}
Explicitly, equation \eqref{parallel} satisfies
\begin{equation}\label{parallel1}
\left\{\begin{split}
&\partial_t\gamma+\nabla\cdot(\chi(\rho)\nabla\Phi)=0,\\
&\nabla\cdot\Big(\chi(\rho)\nabla \big[\partial_t\eta+\frac{1}{2}(\nabla\Phi, \chi'(\rho)\nabla\eta)\big]-\frac{1}{2}\nabla\cdot(\chi\nabla\Phi)\chi'(\rho)\nabla\eta+\frac{1}{2}\nabla\cdot(\chi\nabla\eta)\chi'(\rho)\nabla\Phi\Big)=0.
\end{split}\right.
\end{equation}
In addition, the following statements hold:
\begin{itemize}
\item[(i)] If $\eta_1(t)$, $\eta_2(t)$ is parallel along the curve $\gamma(t)$, then  
\begin{equation*}
\frac{d}{dt}\langle \V_{\eta_1}, \V_{\eta_2}\rangle=0.
\end{equation*}
\item[(ii)] The geodesic equation satisfies  
\begin{equation*}
\partial_t\gamma+ \Delta_{\chi}\Phi=0,\quad \Delta_{\chi}\big(\partial_t\Phi+ \frac{1}{2}\Gamma_{\chi'}(\Phi, \Phi)\big)=0.
\end{equation*}
Explicitly, a geodesics equation forms a {coupled system of continuity equation and Hamilton-Jacobi equation with nonlinear mobilities}  
\begin{equation}\label{geo}
\left\{\begin{split}
& \partial_t\gamma(t)+\nabla\cdot(\chi(\rho(t))\nabla \Phi(t))=0,\\
&\partial_t \Phi(t)+\frac{1}{2}(\nabla \Phi(t), \chi'(\rho(t))\nabla\Phi(t))=0.
\end{split}\right.
\end{equation}
\end{itemize}
\end{theorem}
\begin{proof}
By the Levi--Civita connection \eqref{V12}, we derive \eqref{parallel}. (i). Since $\langle \V_{\eta_1}, \V_{\eta_2}\rangle=\int \eta_1 (-\Delta_{\chi}\eta_2)dx$, and $\eta_1$, $\eta_2$ satisfy \eqref{parallel}, then 
\begin{equation*}
\begin{split}
\frac{d}{dt}\langle \V_{\eta_1}, \V_{\eta_2}\rangle
=&\int (-\Delta_{\chi}\partial_t\eta_1) \eta_2 dx-\int \eta_1 \Delta_{\V_{\Phi}\chi}\eta_2 dx+\int \eta_1 (-\Delta_{\chi} \partial_t\eta_2) dx\\
=&-\int\Big\{-\frac{1}{2}\eta_2\Delta_{\V_{\Phi}\chi}\eta_1+\frac{1}{2}\eta_2 \Delta_{\V_{\eta_1}\chi}\Phi-\frac{1}{2}\Phi \Delta_{\V_{\eta_2}\chi}\eta_1\\
&\hspace{1.2cm}+\eta_1 \Delta_{\V_{\Phi}\chi}\eta_2-\frac{1}{2}\eta_1\Delta_{\V_{\Phi}\chi}\eta_2+\frac{1}{2}\eta_1 \Delta_{\V_{\eta_2}\chi}\Phi-\frac{1}{2}\Phi\Delta_{\V_{\eta_1}\chi}\eta_2\Big\}dx\\
=&0.
\end{split}
\end{equation*}
(ii). If $\eta=\Phi$, then $(\gamma, \eta)$ satisfies the geodesic equation:  
$$ \V_{\partial_t\Phi}+\bar\nabla_{\V_{\Phi}}\V_{\Phi}=0.$$
Since $[\V_{\Phi}, \V_{\Phi}]=0$, then the geodesic equation satisfies
\begin{equation*}
(-\Delta_{\chi})\Big(\partial_t\Phi+\frac{1}{2}\Gamma_{\chi'}(\Phi, \Phi)\Big)=0.
\end{equation*}
Thus, there exists a scalar function $c(t)$, such that 
\begin{equation*}
\partial_t\Phi+\frac{1}{2}\Gamma_{\chi'}(\Phi, \Phi)=c(t).
\end{equation*}
{Clearly, if one defines $\tilde \Phi=\Phi+\int_0^tc(s)ds$, then $(\rho, \tilde\Phi)$ also satisfies the geodesic equation, with $\partial_t\tilde\Phi+\frac{1}{2}\Gamma_{\chi'}(\tilde\Phi,\tilde\Phi)=0$.}
This finishes the proof. 
\end{proof}
\begin{remark}
{We present an example of equation \eqref{geo}, which agrees with the one derived in \cite[equation (2.3)]{C1} and \cite{O1}.  
If $\chi(\rho)=m(\rho)\mathbb{I}$, $m(\rho)>0$ is a smooth function, then the geodesic in \eqref{geo} satisfies 
\begin{equation*}
\left\{\begin{split}
& \partial_t\gamma(t)+\nabla\cdot\big(m(\rho(t))\nabla \Phi(t)\big)=0,\\
&\partial_t \Phi(t)+\frac{1}{2}\big(\nabla \Phi(t), \nabla\Phi(t)\big)m'(\rho(t))=0.
\end{split}\right.
\end{equation*}
}
\end{remark}

We last present the Hessian operator of {an energy functional} $F$ in $(\mathcal{P}_+, \g)$. 
\begin{lemma}\label{diff}
Given a functional $F\in C^{\infty}(\sP_+; \mathbb{R})$, denote the Hessian operator of $F$ in $(\mathcal{P}_+, \g)$ as $\bar\Hess F:=\Hess_{\g}F\colon \sP_+\times C^{\infty}(\Omega)\times C^{\infty}(\Omega)\rightarrow\mathbb{R}$. 
Then the Hessian operator of $F$ at directions $\V_{\Phi_1}$, $\V_{\Phi_2}$ satisfies  
\begin{equation}\label{Hess}
\begin{split}
& \bar{\Hess}F(\rho)\langle \V_{\Phi_1}, \V_{\Phi_2}\rangle\\=&
\quad\int \int \Big(  \nabla_x\nabla_y \frac{\delta^2}{\delta\rho^2}F(\rho)(x,y)\chi(\rho(x))\nabla_x\Phi_1(x), \chi(\rho(y))\nabla_y\Phi_2(y)\Big) dxdy
\\
&+\frac{1}{2}\int\Big\{
\Gamma_{\chi}(\Gamma_{\chi'}(\Phi_2, \frac{\delta}{\delta\rho}F(\rho)), \Phi_1)+\Gamma_{\chi}(\Gamma_{\chi'}(\Phi_1, \frac{\delta}{\delta\rho}F(\rho)), \Phi_2)
-\Gamma_{\chi}(\Gamma_{\chi'}(\Phi_1, \Phi_2), \frac{\delta}{\delta\rho}F(\rho))\Big\}dx.
 \end{split}
 \end{equation}
 Here $\frac{\delta^2}{\delta\rho^2}$ is the $L^2$ second order variation formula. 
 \end{lemma}
\begin{proof}
From the definition of the Hessian operator, we have 
\begin{equation*}
\begin{split}
&\bar{\Hess}F(\rho)\langle \V_{\Phi_1}, \V_{\Phi_2}\rangle\\
=&\V_{\Phi_1}\langle \V_{\mathrm{grad}F}, \V_{\Phi_2}\rangle- \langle \V_{\mathrm{grad}F}, \bar\nabla_{\V_{\Phi_1}}\V_{\Phi_2}\rangle\\
=&\int \int \frac{\delta^2}{\delta\rho^2}F(\rho)(x,y) (-\Delta_{\chi}\Phi_1)(x)(-\Delta_{\chi}\Phi_2)(y)dxdy+\int \frac{\delta}{\delta \rho}F(\rho) (-\Delta_{\V_{\Phi_1}\chi}\Phi_2-\bar\nabla_{\V_{\Phi_1}}\V_{\Phi_2}) dx. \end{split}
\end{equation*}
Applying the explicit formula of the Levi--Civita connection in Lemma \ref{lemma4}, we have 
\begin{equation*}
\begin{split}
&\int \frac{\delta}{\delta \rho}F(\rho)\cdot \left(-\Delta_{\V_{\Phi_1}\chi}\Phi_2-\bar\nabla_{\V_{\Phi_1}}\V_{\Phi_2}\right) dx\\
=&\int \frac{\delta}{\delta \rho}F(\rho)\cdot \left(-\Delta_{\V_{\Phi_1}\chi}\Phi_2+\frac{1}{2}\Big\{\Delta_{{{\V_{\Phi_1}}\chi}}\Phi_2-\Delta_{{{\V_{\Phi_2}}\chi}}\Phi_1+\Delta_{\chi}\Gamma_{\chi'}(\Phi_1,\Phi_2)\Big\}\right) dx\\
=&\frac{1}{2}\int \frac{\delta}{\delta \rho}F(\rho) \cdot\Big\{-\Delta_{{{\V_{\Phi_1}}\chi}}\Phi_2-\Delta_{{{\V_{\Phi_2}}\chi}}\Phi_1+\Delta_{\chi}\Gamma_{\chi'}(\Phi_1,\Phi_2)\Big\} dx.
\end{split}
\end{equation*}
From the integration by parts in above formulas, we finish the proof.
\end{proof}
\begin{remark}
There are several examples of Hessian operators for different energy functionals $F$ in $(\mathcal{P}_+, \g)$. Typical choices of energy functionals $F$ include linear, interaction potential energies, and entropies, see details in \cite{C1,O1, LiG2}. {We use an example to demonstrate that the proposed Hessian operator agrees with the ones in \cite{C1}. Consider a linear energy $\mathcal{F}(\rho)=\int V(x)\rho dx$, where $V\in C^{\infty}(\Omega; \mathbb{R})$ is a smooth potential function. By letting $\Phi_1=\Phi_2$, and conducting some computations, we have 
\begin{equation*}
\begin{split}
 \bar{\Hess}F(\rho)\langle \V_{\Phi}, \V_{\Phi}\rangle=&\int\Big\{
\Gamma_{\chi}(\Gamma_{\chi'}(\Phi, V), \Phi)
-\frac{1}{2}\Gamma_{\chi}(\Gamma_{\chi'}(\Phi, \Phi), V)\Big\}dx\\
=&\int \Big\{(\nabla^2V\nabla\Phi, \nabla\Phi)m(\rho)m'(\rho)\\
&\qquad+ \Big[(\nabla \Phi, \nabla\rho)(\nabla\Phi, \nabla V)-\frac{1}{2}(\nabla \rho, \nabla V)(\nabla\Phi, \nabla\Phi)\Big]m(\rho)m''(\rho)\Big\}dx. 
 \end{split}
 \end{equation*}
The above computation agrees with the one in \cite[section2.3]{C1}.
 }
\end{remark}

\subsection{Riemannian curvature tensor}
In this section, we present the main result of this paper. We derive the Rimennain curvature tensor in $(\mathcal{P}_+, \g)$. Denote $\bar{\R}=\R_\g\colon \sP_+\times C^{\infty}(\Omega)\times C^{\infty}(\Omega)\times C^{\infty}(\Omega)\rightarrow C^{\infty}(\Omega)$.
\begin{theorem}[Riemannian curvature in hydrodynamical density manifold]\label{theorem2}
Given functions $\Phi_1$, $\Phi_2$, $\Phi_3$, $\Phi_4\in C^{\infty}(\Omega)$, the Riemannian curvature in $(\sP_+, \g)$ at directions $\V_{\Phi_1}$, $\V_{\Phi_2}$, $\V_{\Phi_3}$, $\V_{\Phi_4}$ satisfies 
\begin{equation}\label{tensor}
\begin{split}
&\langle \bar \R(\V_{\Phi_1}, \V_{\Phi_2})\V_{\Phi_3}, \V_{\Phi_4}\rangle\\
=&\quad\frac{1}{2}\int \Big\{-\Gamma_{\chi''}(\Phi_2, \Phi_4)\cdot\Delta_{\chi}\Phi_1\cdot\Delta_{\chi}\Phi_3-\Gamma_{\chi''}(\Phi_1, \Phi_3)\cdot\Delta_{\chi}\Phi_2\cdot\Delta_{\chi}\Phi_4\\
&\hspace{1.5cm}+\Gamma_{\chi''}(\Phi_2, \Phi_3)\cdot\Delta_{\chi}\Phi_1\cdot\Delta_{\chi}\Phi_4+\Gamma_{\chi''}(\Phi_1, \Phi_4)\cdot\Delta_{\chi}\Phi_2\cdot\Delta_{\chi}\Phi_3\Big\}dx\\
&+\frac{1}{4}\int\Big\{-\Gamma_{\chi}(\Gamma_{\chi'}(\Gamma_{\chi'}(\Phi_2,\Phi_4), \Phi_1), \Phi_3)-\Gamma_{\chi}(\Gamma_{\chi'}(\Gamma_{\chi'}(\Phi_2,\Phi_4), \Phi_3), \Phi_1)\\
&\hspace{1.5cm}-\Gamma_{\chi}(\Gamma_{\chi'}(\Gamma_{\chi'}(\Phi_1,\Phi_3), \Phi_2), \Phi_4)-\Gamma_{\chi}(\Gamma_{\chi'}(\Gamma_{\chi'}(\Phi_1,\Phi_3), \Phi_4), \Phi_2)\\
&\hspace{1.5cm}+\Gamma_{\chi}(\Gamma_{\chi'}(\Gamma_{\chi'}(\Phi_2,\Phi_3), \Phi_1), \Phi_4)+\Gamma_{\chi}(\Gamma_{\chi'}(\Gamma_{\chi'}(\Phi_2,\Phi_3), \Phi_4), \Phi_1)\\
&\hspace{1.5cm}+\Gamma_{\chi}(\Gamma_{\chi'}(\Gamma_{\chi'}(\Phi_1,\Phi_4), \Phi_2), \Phi_3)+\Gamma_{\chi}(\Gamma_{\chi'}(\Gamma_{\chi'}(\Phi_1,\Phi_4), \Phi_3), \Phi_2)\\
&\hspace{1.5cm}+ \Gamma_{\chi}(\Gamma_{\chi'}(\Phi_1, \Phi_3), \Gamma_{\chi'}(\Phi_2, \Phi_4))- \Gamma_{\chi}(\Gamma_{\chi'}(\Phi_2, \Phi_3), \Gamma_{\chi'}(\Phi_1, \Phi_4))\Big\}dx\\
&-\frac{1}{4}\int\Big\{[\V_{\Phi_1}, \V_{\Phi_3}]\cdot\Delta_{\chi}^{\dd}[\V_{\Phi_2}, \V_{\Phi_4}]-[\V_{\Phi_2}, \V_{\Phi_3}]\cdot\Delta_{\chi}^{\dd}[\V_{\Phi_1}, \V_{\Phi_4}]+2[\V_{\Phi_3}, \V_{\Phi_4}]\cdot\Delta_{\chi}^{\dd}[\V_{\Phi_1}, \V_{\Phi_2}]\Big\} dx.
\end{split}
\end{equation}
\end{theorem}
\begin{proof}
To derive the curvature tensor of vector fields $\V_{\Phi_a}$, $a=1, 2,3,4$,  we apply the following formula:
\begin{equation}\label{main}
\begin{split}
\langle \bar\R(\V_{\Phi_1}, \V_{\Phi_2}) \V_{\Phi_3}, \V_{\Phi_4}\rangle=& \langle \bar\nabla_{\V_{\Phi_1}}\bar\nabla_{\V_{\Phi_2}}\V_{\Phi_3}-\bar\nabla_{\V_{\Phi_2}}\bar\nabla_{\V_{\Phi_1}}\V_{\Phi_3}-\bar\nabla_{[\V_{\Phi_1}, \V_{\Phi_2}]}\V_{\Phi_3} , \V_{\Phi_4} \rangle\\
=&\quad \V_{\Phi_1}\langle \bar\nabla_{\V_{\Phi_2}}\V_{\Phi_3}, \V_{\Phi_4}\rangle - \langle \bar\nabla_{\V_{\Phi_2}}\V_{\Phi_3}, \bar\nabla_{\V_{\Phi_1}}\V_{\Phi_4}\rangle \\
&- \V_{\Phi_2}\langle\bar\nabla_{\V_{\Phi_1}}\V_{\Phi_3},\V_{\Phi_4}\rangle+\langle\bar\nabla_{\V_{\Phi_1}}\V_{\Phi_3}, \bar\nabla_{\V_{\Phi_2}}\V_{\Phi_4}\rangle\\
&- \langle\bar\nabla_{[\V_{\Phi_1}, \V_{\Phi_2}]}\V_{\Phi_3}, \V_{\Phi_4}\rangle.
\end{split}
\end{equation}
We estimate the above formula in three steps. Firstly, for indices $a,b,c\in \{1,2,3,4\}$, from \eqref{a}, we denote
\begin{equation*}
\V_{abc}(\rho):=\langle \bar\nabla_{\V_{\Phi_a}}\V_{\Phi_b}, \V_{\Phi_c}\rangle =-\frac{1}{2}\int \Big(\Phi_c \Delta_{\V_{\Phi_a}\chi}\Phi_b- \Phi_c\Delta_{\V_{\Phi_b}\chi}\Phi_a+\Phi_a \Delta_{\V_{\Phi_c}\chi}\Phi_b\Big)dx,
\end{equation*}
and write $\V_{\Phi_a}\V_{\Phi_b}\chi:=((\V_{\Phi_a}\V_{\Phi_b}\chi)_{ij})_{1\leq i,j\leq d}$, with 
\begin{equation}\label{VaVb}
(\V_{\Phi_a}\V_{\Phi_b}\chi)_{ij}=\nabla\cdot\Big(\chi'\nabla\cdot(\chi\nabla\Phi_a)\nabla\Phi_b\Big)\frac{\partial}{\partial \rho}\chi_{ij}(\rho)+\frac{\partial^2}{\partial\rho^2}\chi_{ij}(\rho)\cdot\Delta_{\chi}\Phi_a\Delta_{\chi}\Phi_b. 
\end{equation}
Thus
\begin{equation}\label{ref1}
\begin{split}
\V_{\Phi_1}\langle \bar\nabla_{\V_{\Phi_2}}\V_{\Phi_3}, \V_{\Phi_4}\rangle
=&\frac{d}{d\epsilon}|_{\epsilon=0} \V_{234}(\rho-\epsilon\Delta_{\chi}\Phi_1) \\
=&-\frac{1}{2}\int \Big(\Phi_4 \cdot\Delta_{\V_{\Phi_1}\V_{\Phi_2}\chi}\Phi_3-\Phi_4 \cdot\Delta_{\V_{\Phi_1}\V_{\Phi_3}\chi}\Phi_2+\Phi_2\cdot\Delta_{\V_{\Phi_1}\V_{\Phi_4}\chi}\Phi_3\Big)dx. 
\end{split}
\end{equation}
Similarly, by exchanging the index $1$, $2$, we have 
\begin{equation}\label{ref2}
\begin{split}
\V_{\Phi_2}\langle \bar\nabla_{\V_{\Phi_1}}\V_{\Phi_3}, \V_{\Phi_4}\rangle 
=&-\frac{1}{2}\int \Big(\Phi_4 \cdot\Delta_{\V_{\Phi_2}\V_{\Phi_1}\chi}\Phi_3-\Phi_4 \cdot\Delta_{\V_{\Phi_2}\V_{\Phi_3}\chi}\Phi_1+\Phi_1\cdot\Delta_{\V_{\Phi_2}\V_{\Phi_4}\chi}\Phi_3\Big)dx. 
\end{split}
\end{equation}
Secondly, since 
\begin{equation*}
\bar\nabla_{\V_{\Phi_2}}\V_{\Phi_3}=\frac{1}{2}[\V_{\Phi_2},\V_{\Phi_3}]-\frac{1}{2}\Delta_{\chi}\Gamma_{\chi'}(\Phi_2, \Phi_3).
\end{equation*}
Then 
\begin{equation}\label{ref3}
\begin{split}
\langle \bar\nabla_{\V_{\Phi_2}}\V_{\Phi_3}, \bar\nabla_{\V_{\Phi_1}}\V_{\Phi_4}\rangle=&\int \bar\nabla_{\V_{\Phi_2}}\V_{\Phi_3} \cdot (-\Delta_{\chi})^{\dd} \bar\nabla_{\V_{\Phi_1}}\V_{\Phi_4}dx\\
=&\frac{1}{4}\int \Big\{-[\V_{\Phi_2}, \V_{\Phi_3}]\cdot\Delta_{\chi}^{\dd}[\V_{\Phi_1}, \V_{\Phi_4}]+[\V_{\Phi_2}, \V_{\Phi_3}]\cdot\Gamma_{\chi'}(\Phi_1,\Phi_4)\\
&\hspace{1cm}+[\V_{\Phi_1}, \V_{\Phi_4}]\cdot\Gamma_{\chi'}(\Phi_2, \Phi_3)-\Gamma_{\chi'}(\Phi_2,\Phi_3)\cdot\Delta_{\chi}\Gamma_{\chi'}(\Phi_1, \Phi_4)\Big\}dx.
\end{split}
\end{equation}
Similarly, by exchanging the index $1$, $2$, we have \begin{equation}\label{ref4}
\begin{split}
\langle \bar\nabla_{\V_{\Phi_1}}\V_{\Phi_3}, \bar\nabla_{\V_{\Phi_2}}\V_{\Phi_4}\rangle
=&\frac{1}{4}\int\Big\{-[\V_{\Phi_1}, \V_{\Phi_3}]\cdot\Delta_{\chi}^{\dd}[\V_{\Phi_2}, \V_{\Phi_4}]+[\V_{\Phi_1}, \V_{\Phi_3}]\cdot\Gamma_{\chi'}(\Phi_2,\Phi_4)\\
&\hspace{1cm}+[\V_{\Phi_2}, \V_{\Phi_4}]\cdot\Gamma_{\chi'}(\Phi_1, \Phi_3)-\Gamma_{\chi'}(\Phi_1,\Phi_3)\cdot\Delta_{\chi}\Gamma_{\chi'}(\Phi_2, \Phi_4)\Big\}dx.
\end{split}
\end{equation}

Thirdly, denote $\V_{\Phi}=[\V_{\Phi_1}, \V_{\Phi_2}]$, where
$$\Phi=-\Delta_{\chi}^{\dd}[\V_{\Phi_1}, \V_{\Phi_2}].$$
Then 
\begin{equation}\label{ref5}
\begin{split}
&\langle \bar\nabla_{[\V_{\Phi_1}, \V_{\Phi_2}]} \V_{\Phi_3}, \V_{\Phi_4}\rangle\\
=&\frac{1}{2}\int \Phi_4\cdot\Big\{-\Delta_{\V_\Phi\chi}\Phi_3+\Delta_{\V_{\Phi_3}\chi}\Phi-\Delta_{\chi}\Gamma_{\chi'}(\Phi, \Phi_3)\Big\}dx\\
=&-\int \frac{1}{2}\Phi_4 \cdot\Delta_{\V_\Phi\chi}\Phi_3-\frac{1}{2}\Phi_4 \cdot\Delta_{\V_{\Phi_3}\chi}\Phi+\frac{1}{2}\Phi_3\cdot\Delta_{\V_{\Phi_4}\chi}\Phi dx\\
=&\int -\frac{1}{2}\Phi_4 \cdot\Delta_{[\V_{\Phi_1}, \V_{\Phi_2}]\chi}\Phi_3-\frac{1}{2}[\V_{\Phi_4}, \V_{\Phi_3}]\cdot\Delta_{\chi}^{\dd}[\V_{\Phi_1}, \V_{\Phi_2}] dx\\
=&\int -\frac{1}{2}\Phi_4 \cdot\Delta_{\V_{\Phi_1}\V_{\Phi_2}\chi}\Phi_3+\frac{1}{2}\Phi_4 \cdot\Delta_{\V_{\Phi_2}\V_{\Phi_1}\chi}\Phi_3-\frac{1}{2}[\V_{\Phi_4}, \V_{\Phi_3}]\cdot\Delta_{\chi}^{\dd}[\V_{\Phi_1}, \V_{\Phi_2}]dx.
\end{split}
\end{equation}
For the simplicity of presentation, we write 
\begin{equation}\label{abcd}
\begin{split}
(abcd):=&-\int \Phi_a\cdot\Delta_{\V_{\Phi_b}\V_{\Phi_c}\chi}\Phi_d dx\\
=&\int \Gamma_{\chi}(\Gamma_{\chi'}(\Gamma_{\chi'}(\Phi_a,\Phi_d), \Phi_c), \Phi_b)+\Gamma_{\chi''}(\Phi_a,\Phi_d)\cdot\Delta_{\chi}\Phi_b\Delta_{\chi}\Phi_c dx, 
\end{split}
\end{equation}
where the second equation is derived from the definition \eqref{VaVb} and the integration by parts. 
Clearly, we have $(abcd)=(dbca)$. Substituting equations \eqref{ref1}, \eqref{ref2}, \eqref{ref3}, \eqref{ref4}, \eqref{ref5} into equation \eqref{main}, we have 
\begin{equation}\label{ref6}
\begin{split}
&\langle \bar\R(\V_{\Phi_1}, \V_{\Phi_2})\V_{\Phi_3}, \V_{\Phi_4}\rangle\\
=&\quad\frac{1}{4}\big\{(2143)+(2413)\big\}-\frac{1}{4}\big\{(2134)+(2314)\big\}+\frac{1}{4}\big\{(1234)+(1324)\big\}-\frac{1}{4}\big\{(1243)+(1423)\big\}\\ 
&-\frac{1}{4}\int\Big\{\Gamma_{\chi'}(\Phi_1, \Phi_3)\cdot\Delta_{\chi}\Gamma_{\chi'}(\Phi_2, \Phi_4)+\Gamma_{\chi'}(\Phi_2, \Phi_3)\cdot\Delta_{\chi}\Gamma_{\chi'}(\Phi_1, \Phi_4)\Big\} dx\\
&-\frac{1}{4}\int\Big\{ [\V_{\Phi_1}, \V_{\Phi_3}]\cdot\Delta_{\chi}^{\dd}[\V_{\Phi_2}, \V_{\Phi_4}]-[\V_{\Phi_2}, \V_{\Phi_3}]\cdot\Delta_{\chi}^{\dd}[\V_{\Phi_1}, \V_{\Phi_4}]+2  [\V_{\Phi_3}, \V_{\Phi_4}]\cdot\Delta_{\chi}^{\dd}[\V_{\Phi_1}, \V_{\Phi_2}]\Big\}dx.
\end{split}
\end{equation}
From equation \eqref{abcd} and the integration by parts, we derive equation \eqref{tensor}. This finishes the proof.
\end{proof}
We last compute the sectional curvature in $(\sP_+, \g)$. Denote $\bar{\K}=\K_\g\colon C^{\infty}(\Omega)\times C^{\infty}(\Omega)\rightarrow \mathbb{R}$.
\begin{corollary}[Sectional curvature in hydrodynamical density manifold]
Given functions $\Phi_1$, $\Phi_2\in C^{\infty}(\Omega)$, the sectional curvature in $(\sP_+, \g)$ at directions $\V_{\Phi_1}$, $\V_{\Phi_2}$ satisfies 
\begin{equation}\label{sec_tensor}
\begin{split}
&\bar\K(\V_{\Phi_1}, \V_{\Phi_2})\\
=&\quad\frac{1}{2Z}\int \Big\{-2\Gamma_{\chi''}(\Phi_1, \Phi_2)\cdot\Delta_{\chi}\Phi_1\cdot\Delta_{\chi}\Phi_2+\Gamma_{\chi''}(\Phi_2, \Phi_2)(\Delta_{\chi}\Phi_1)^2+\Gamma_{\chi''}(\Phi_1, \Phi_1)(\Delta_{\chi}\Phi_2)^2\Big\}dx\\
&+\frac{1}{4Z}\int\Big\{-2\Gamma_{\chi}(\Gamma_{\chi'}(\Gamma_{\chi'}(\Phi_1,\Phi_2), \Phi_1), \Phi_2)-2\Gamma_{\chi}(\Gamma_{\chi'}(\Gamma_{\chi'}(\Phi_1,\Phi_2), \Phi_2), \Phi_1)\\
&\hspace{1.9cm}+2\Gamma_{\chi}(\Gamma_{\chi'}(\Gamma_{\chi'}(\Phi_1,\Phi_1), \Phi_2), \Phi_2)+2\Gamma_{\chi}(\Gamma_{\chi'}(\Gamma_{\chi'}(\Phi_2,\Phi_2), \Phi_1), \Phi_1)\\
&\hspace{1.9cm}+ \Gamma_{\chi}(\Gamma_{\chi'}(\Phi_1, \Phi_2), \Gamma_{\chi'}(\Phi_1, \Phi_2))- \Gamma_{\chi}(\Gamma_{\chi'}(\Phi_1, \Phi_1), \Gamma_{\chi'}(\Phi_2, \Phi_2))\Big\}dx\\
&+\frac{3}{4Z}\int [\V_{\Phi_1}, \V_{\Phi_2}]\cdot\Delta_{\chi}^{\dd}[\V_{\Phi_1}, \V_{\Phi_2}]dx,
\end{split}
\end{equation}
where $Z\in \mathbb{R}_+$ is a scalar defined as
\begin{equation}\label{Z}
\begin{split}
Z:=&\int \Gamma_\chi(\Phi_1, \Phi_1)dx\cdot \int \Gamma_{\chi}(\Phi_2, \Phi_2)dx- |\int \Gamma_{\chi}(\Phi_1, \Phi_2)dx|^2. 
\end{split}
\end{equation}
\end{corollary}
\begin{proof}
Based on the definition of sectional curvature, we have 
\begin{equation*}
{\bar\K}(\V_{\Phi_1},\V_{\Phi_2})=\frac{\bar{\R}((\V_{\Phi_1},\V_{\Phi_2})\V_{\Phi_2}, \V_{\Phi_1})}{\langle \V_{\Phi_1}, \V_{\Phi_1}\rangle\langle \V_{\Phi_2}, \V_{\Phi_2}\rangle- \langle \V_{\Phi_1}, \V_{\Phi_2}\rangle^2}.
\end{equation*}
From the fact that $[\V_{\Phi_i}, \V_{\Phi_i}]=0$ with $i=1,2$, we finish the proof. 
\end{proof}
\begin{remark}{
From the Cauchy-Schwarz inequality and $\chi(\rho)\geq 0$, it is clear that $Z$ defined in \eqref{Z} is non-negative, i.e., 
\begin{equation*}
(\int (\nabla\Phi_1, \chi(\rho)\nabla\Phi_2) dx)^2\leq \int (\nabla\Phi_1, \chi(\rho)\nabla\Phi_1) dx\cdot \int (\nabla\Phi_2, \chi(\rho)\nabla\Phi_2) dx. 
\end{equation*}}
\end{remark}

\section{Curvatures in one dimensional density space}\label{sec4}
{When $\Omega=\mathbb{R}^1$}, we derive some explicit formulas for Riemannian and sectional curvatures in hydrodynamical density manifolds.  
\begin{theorem}\label{thm3}
Suppose $\Omega=\mathbb{R}^1$. Given functions $\Phi_1$, $\Phi_2$, $\Phi_3$, $\Phi_4\in C^{\infty}(\Omega)$, 
the Riemannian curvature in $(\sP_+, \g)$ at directions $\V_{\Phi_1}$, $\V_{\Phi_2}$, $\V_{\Phi_3}$, $\V_{\Phi_4}$ satisfies 
\begin{equation}\label{Riem_tensor_exp}
\begin{split}
&\langle \bar \R(\V_{\Phi_1}, \V_{\Phi_2})\V_{\Phi_3}, \V_{\Phi_4}\rangle\\
=&\quad\frac{1}{2}\int \chi''(\rho(x))\chi(\rho(x))^2\Big\{-\Phi_2'(x)\Phi_4'(x)\Phi_1''(x)\Phi_3''(x)-\Phi_1'(x)\Phi_3'(x)\Phi_2''(x)\Phi_4''(x)\\
&\hspace{4.4cm}+\Phi'_2(x)\Phi'_3(x)\Phi''_1(x)\Phi''_4(x)+\Phi'_1(x)\Phi'_4(x)\Phi''_2(x)\Phi''_3(x)\Big\}dx.\end{split}
\end{equation}
Moreover, the sectional curvature at directions $\V_{\Phi_1}$, $\V_{\Phi_2}$ satisfies 
\begin{equation}\label{sec_tensor_exp}
\begin{split}
\bar\K(\V_{\Phi_1}, \V_{\Phi_2})
=&\frac{1}{2Z}\int \chi''(\rho(x))\chi(\rho(x))^2\big(\Phi''_2(x)\Phi'_1(x)-\Phi''_1(x)\Phi'_2(x)\big)^2dx,
\end{split}
\end{equation}
where $Z\in\mathbb{R}_+$ is a scalar defined in equation \eqref{Z}, such that 
\begin{equation*}
Z=\int |\Phi'_1|^2\chi(\rho)dx \cdot \int |\Phi'_2|^2\chi(\rho)dx -|\int \Phi'_1\Phi'_2\chi(\rho) dx|^2. 
\end{equation*}
If $\chi(\rho)$ is convex in $\rho$, then the sectional curvature $\bar\K$ is nonnegative. If $\chi(\rho)$ is concave in $\rho$, then the sectional curvature $\bar\K$ is nonpositive.   
\end{theorem}

The proof of Theorem \ref{thm3} is based on some calculations and cancellations. We need to prove the following lemmas. Denote $\chi'=\chi'(\rho(x))$, $\partial\chi=\partial_x\chi(\rho(x))$, $\partial \chi'=\partial_x(\chi'(\rho(x)))$, and  $\partial \chi'^2=\partial_x(\chi'(\rho(x))^2)$. 
Denote indices $a$, $b$, $c$, $d\in\{1,2,3,4\}$.
\begin{lemma}\label{lemma6}
If $\Omega=\mathbb{R}^1$, then the following equalities hold. \begin{itemize} 
\item[(i)]
\begin{equation*}
\begin{split}
&\Gamma_{\chi}(\Gamma_{\chi'}(\Phi_a, \Phi_b),\Gamma_{\chi'}(\Phi_c, \Phi_d))\\=&\chi\chi'^2\partial\Gamma_1(\Phi_a,\Phi_b)\partial\Gamma_1(\Phi_c,\Phi_d)
+\frac{1}{2}\chi\partial\chi'^2\partial(\Gamma_1(\Phi_c,\Phi_d) \Gamma_1(\Phi_a, \Phi_b))+\chi(\partial \chi')^2\Gamma_1(\Phi_c,\Phi_d)\Gamma_1(\Phi_a, \Phi_b). 
\end{split}
\end{equation*}
\item[(ii)]
\begin{equation*}
\begin{split}
&\Gamma_{\chi}(\Gamma_{\chi'}(\Gamma_{\chi'}(\Phi_a, \Phi_b), \Phi_c), \Phi_d)\\=&\chi\chi'^2\Gamma_1(\Gamma_1(\Gamma_1(\Phi_a,\Phi_b), \Phi_c),\Phi_d)
+\chi\partial\chi'^2\Gamma_1(\Gamma_1(\Phi_a,\Phi_b), \Phi_c) \Phi'_d\\
&+\chi \chi'\partial\chi'(\Phi_a'\Phi_b'\Phi_c')'\Phi'_d+\chi\big((\partial\chi')^2+\chi'\partial^2\chi'\big)\Phi_a'\Phi_b'\Phi_c'\Phi_d'. 
\end{split}
\end{equation*}
\item[(iii)]
\begin{equation*}
\begin{split}
\Gamma_{\chi''}(\Phi_a, \Phi_b)\cdot\Delta_{\chi}\Phi_c \Delta_{\chi}\Phi_d=\chi'' \Phi'_a\Phi'_b (\partial\chi\Phi'_c+\chi \Phi_c'') (\partial\chi\Phi'_d+\chi \Phi_d'').
\end{split}
\end{equation*}
\end{itemize}
\end{lemma}
\begin{proof}
(i) From the product rule, we have 
\begin{equation*}
\begin{split}
&\Gamma_{\chi}(\Gamma_{\chi'}(\Phi_a, \Phi_b),\Gamma_{\chi'}(\Phi_c, \Phi_d))\\
=&\chi\partial (\Gamma_1(\Phi_a, \Phi_b) \chi') \partial (\Gamma_1(\Phi_c, \Phi_d) \chi')\\ 
=&\chi (\partial\Gamma_1(\Phi_a, \Phi_b) \chi'+\Gamma_1(\Phi_a, \Phi_b)\partial\chi')  (\partial\Gamma_1(\Phi_c, \Phi_d) \chi'+\Gamma_1(\Phi_c, \Phi_d) \partial\chi')\\ 
=&\chi\chi'^2\partial\Gamma_1(\Phi_a,\Phi_b)\partial\Gamma_1(\Phi_c,\Phi_d)
+\chi\chi'\partial\chi'\partial(\Gamma_1(\Phi_c,\Phi_d) \Gamma_1(\Phi_a, \Phi_b))+\Gamma_1(\Phi_c,\Phi_d)\Gamma_1(\Phi_a, \Phi_b)\chi(\partial \chi')^2. 
\end{split}
\end{equation*}
Using $\chi'\partial \chi'=\frac{1}{2}\partial \chi'^2$, we finish the proof. 

\noindent(ii) Again, from the product rule, we have
\begin{equation*}
\begin{split}
&\Gamma_{\chi}(\Gamma_{\chi'}(\Gamma_{\chi'}(\Phi_a, \Phi_b), \Phi_c), \Phi_d)\\
=&\chi \partial(\partial(\Phi'_a\Phi'_b\chi') \Phi'_c\chi') \Phi'_d\\
=&\chi \partial(\partial(\Phi'_a\Phi'_b) \Phi'_c\chi'^2) \Phi'_d+\chi \partial(\Phi'_a\Phi'_b \Phi'_c\chi'\partial\chi') \Phi'_d\\
=&\chi'^2\chi \partial(\partial(\Phi'_a\Phi'_b) \Phi'_c) \Phi'_d+\chi\partial\chi'^2\partial(\Phi'_a\Phi'_b) \Phi'_c\Phi'_d+\chi \partial(\Phi'_a\Phi'_b \Phi'_c\chi'\partial\chi') \Phi'_d\\
=&\chi\chi'^2\Gamma_1(\Gamma_1(\Gamma_1(\Phi_a,\Phi_b), \Phi_c),\Phi_d)
+\chi\partial\chi'^2\Gamma_1(\Gamma_1(\Phi_a,\Phi_b), \Phi_c) \Phi'_d+\chi \partial(\Phi_a'\Phi_b'\Phi_c'\chi'\partial\chi')\Phi_d'. 
\end{split}
\end{equation*}
Using the fact that 
\begin{equation*}
\partial(\Phi_a'\Phi_b'\Phi_c'\chi'\partial\chi')=(\Phi_a'\Phi_b'\Phi_c')'\chi'\partial\chi'+\Phi_a'\Phi_b'\Phi_c'\big((\partial\chi')^2+\chi'\partial^2\chi'\big),
\end{equation*}
we finish the proof. 

(iii) It follows from the definition. 
\end{proof}

\begin{lemma}\label{lemma7}
If $\Omega=\mathbb{R}^1$, then the following equality holds. 
\begin{equation*}
-\int [\V_{\Phi_a}, \V_{\Phi_b}]\cdot\Delta_{\chi}^{\dd}[\V_{\Phi_c}, \V_{\Phi_d}]dx=\int \chi\chi'^2(\Phi_a''\Phi'_b-\Phi''_b\Phi'_a)(\Phi_c''\Phi'_d-\Phi''_d\Phi'_c)dx.
\end{equation*}
\end{lemma}
\begin{proof}
We first show that 
\begin{equation*}
 [\V_{\Phi_a}, \V_{\Phi_b}]=-\partial\big(\chi\chi'[\Phi_b''\Phi'_a-\Phi''_a\Phi'_b]\big). 
\end{equation*}
From equation \eqref{comm_ex}, we have 
\begin{equation*}
\begin{split}
-\Delta_{\V_{\Phi_a}\chi}\Phi_b=&\partial\big(\chi' \partial(\chi\Phi_a')\Phi'_b)\\
=&\partial\big(\chi'\partial\chi \Phi'_a\Phi'_b+\chi'\chi\Phi''_a\Phi'_b\big). 
\end{split}
\end{equation*}
Thus, 
\begin{equation*}
[\V_{\Phi_a}, \V_{\Phi_b}]=-\Delta_{\V_{\Phi_a}\chi}\Phi_b+\Delta_{\V_{\Phi_b}\chi}\Phi_a=\partial\big(\chi'\chi[\Phi''_a\Phi'_b-\Phi''_b\Phi'_a]\big).
\end{equation*}
Similarly, 
\begin{equation*}
 [\V_{\Phi_c}, \V_{\Phi_d}]=-\partial\big(\chi\chi'[\Phi_d''\Phi'_c-\Phi''_c\Phi'_d]\big). 
\end{equation*}
Thus, 
\begin{equation*}
\begin{split}
-\int [\V_{\Phi_a}, \V_{\Phi_b}]\cdot\Delta_{\chi}^{\dd}[\V_{\Phi_c}, \V_{\Phi_d}]dx
=&\int \big(\chi\chi'[\Phi_b''\Phi'_a-\Phi''_a\Phi'_b]\big)
\frac{1}{\chi}\big(\chi\chi'[\Phi_d''\Phi'_c-\Phi''_c\Phi'_d]\big)dx\\
=&\int \chi\chi'^2[\Phi_b''\Phi'_a-\Phi''_a\Phi'_b]
[\Phi_d''\Phi'_c-\Phi''_c\Phi'_d]dx, 
\end{split}
\end{equation*}
which finishes the proof. 
\end{proof}

\begin{lemma}\label{lemma8}
If $\Omega=\mathbb{R}^1$, then the following equality holds. 
\begin{equation*}
\begin{split}
&-\Gamma_{1}(\Gamma_{1}(\Gamma_{1}(\Phi_2,\Phi_4), \Phi_1), \Phi_3)-\Gamma_{1}(\Gamma_{1}(\Gamma_{1}(\Phi_2,\Phi_4), \Phi_3), \Phi_1)\\
&-\Gamma_{1}(\Gamma_{1}(\Gamma_{1}(\Phi_1,\Phi_3), \Phi_2), \Phi_4)-\Gamma_{1}(\Gamma_{1}(\Gamma_{1}(\Phi_1,\Phi_3), \Phi_4), \Phi_2)\\
&+\Gamma_{1}(\Gamma_{1}(\Gamma_{1}(\Phi_2,\Phi_3), \Phi_1), \Phi_4)+\Gamma_{1}(\Gamma_{1}(\Gamma_{1}(\Phi_2,\Phi_3), \Phi_4), \Phi_1)\\
&+\Gamma_{1}(\Gamma_{1}(\Gamma_{1}(\Phi_1,\Phi_4), \Phi_2), \Phi_3)+\Gamma_{1}(\Gamma_{1}(\Gamma_{1}(\Phi_1,\Phi_4), \Phi_3), \Phi_2)\\
&+\Gamma_{1}(\Gamma_{1}(\Phi_1, \Phi_3), \Gamma_{1}(\Phi_2, \Phi_4))- \Gamma_{1}(\Gamma_{1}(\Phi_2, \Phi_3), \Gamma_{1}(\Phi_1, \Phi_4))\\
&+(\Phi_3''\Phi'_1-\Phi''_1\Phi'_3)(\Phi_4''\Phi'_2-\Phi''_2\Phi'_4)-(\Phi_3''\Phi'_2-\Phi''_2\Phi'_3)(\Phi_4''\Phi'_1-\Phi''_1\Phi'_4)\\
&+2(\Phi_2''\Phi'_1-\Phi''_1\Phi'_2)(\Phi_4''\Phi'_3-\Phi''_3\Phi'_4)=0.
\end{split}
\end{equation*}
\end{lemma}
\begin{proof}
We first compute that  
\begin{equation*}
\begin{split}
&\Gamma_{1}(\Gamma_{1}(\Gamma_{1}(\Phi_2,\Phi_4), \Phi_1), \Phi_3)+\Gamma_{1}(\Gamma_{1}(\Gamma_{1}(\Phi_2,\Phi_4), \Phi_3), \Phi_1)\\
=&2(\Phi'_2\Phi'_4)''\Phi'_1\Phi'_3+(\Phi'_1\Phi'_3)'(\Phi'_2\Phi'_4)'.
\end{split}
\end{equation*}
By some computations, we have
\begin{equation*}
\begin{split}
&-\Gamma_{1}(\Gamma_{1}(\Gamma_{1}(\Phi_2,\Phi_4), \Phi_1), \Phi_3)-\Gamma_{1}(\Gamma_{1}(\Gamma_{1}(\Phi_2,\Phi_4), \Phi_3), \Phi_1)\\
&-\Gamma_{1}(\Gamma_{1}(\Gamma_{1}(\Phi_1,\Phi_3), \Phi_2), \Phi_4)-\Gamma_{1}(\Gamma_{1}(\Gamma_{1}(\Phi_1,\Phi_3), \Phi_4), \Phi_2)\\
&+\Gamma_{1}(\Gamma_{1}(\Gamma_{1}(\Phi_2,\Phi_3), \Phi_1), \Phi_4)+\Gamma_{1}(\Gamma_{1}(\Gamma_{1}(\Phi_2,\Phi_3), \Phi_4), \Phi_1)\\
&+\Gamma_{1}(\Gamma_{1}(\Gamma_{1}(\Phi_1,\Phi_4), \Phi_2), \Phi_3)+\Gamma_{1}(\Gamma_{1}(\Gamma_{1}(\Phi_1,\Phi_4), \Phi_3), \Phi_2)\\
&+\Gamma_{1}(\Gamma_{1}(\Phi_1, \Phi_3), \Gamma_{1}(\Phi_2, \Phi_4))- \Gamma_{1}(\Gamma_{1}(\Phi_2, \Phi_3), \Gamma_{1}(\Phi_1, \Phi_4))\\
=&3\Big(-\Phi_2''\Phi_4''\Phi_1'\Phi_3'-\Phi_1''\Phi_3''\Phi_2'\Phi_4'+\Phi_2''\Phi_3''\Phi_1'\Phi_4'+\Phi_1''\Phi_4''\Phi_2'\Phi_3'\Big).
\end{split}
\end{equation*}
On the other hand, 
\begin{equation*}
\begin{split}
&\quad (\Phi_3''\Phi'_1-\Phi''_1\Phi'_3)(\Phi_4''\Phi'_2-\Phi''_2\Phi'_4)-(\Phi_3''\Phi'_2-\Phi''_2\Phi'_3)(\Phi_4''\Phi'_1-\Phi''_1\Phi'_4)\\
&+2(\Phi_2''\Phi'_1-\Phi''_1\Phi'_2)(\Phi_4''\Phi'_3-\Phi''_3\Phi'_4)\\
=&3\Big(\Phi_2''\Phi_4''\Phi_1'\Phi_3'+\Phi_1''\Phi_3''\Phi_2'\Phi_4'-\Phi_2''\Phi_3''\Phi_1'\Phi_4'-\Phi_1''\Phi_4''\Phi_2'\Phi_3'\Big). 
\end{split}
\end{equation*}
Summing up the above two formulas, we finish the proof. 
\end{proof}
We are ready to prove Theorem \ref{thm3}. 
\begin{proof}[Proof of Theorem \ref{thm3}]
Using Lemmas \ref{lemma6}, \ref{lemma7} in Theorem \ref{theorem2} and direct calculations, we have
\begin{equation*}
\begin{split}
&\langle \bar \R(\V_{\Phi_1}, \V_{\Phi_2})\V_{\Phi_3}, \V_{\Phi_4}\rangle\\
=&\quad\frac{1}{2}\int  \chi''\chi^2\Big\{-\Phi'_2\Phi'_4\Phi''_1\Phi''_3-\Phi'_1\Phi'_3\Phi''_2\Phi''_4+\Phi'_2\Phi'_3\Phi''_1\Phi''_4+\Phi'_1\Phi'_4\Phi''_2 \Phi''_3\Big\}dx\\
&+\frac{1}{4}\int \chi\chi'^2\Big\{-\Gamma_{1}(\Gamma_{1}(\Gamma_{1}(\Phi_2,\Phi_4), \Phi_1), \Phi_3)-\Gamma_{1}(\Gamma_{1}(\Gamma_{1}(\Phi_2,\Phi_4), \Phi_3), \Phi_1)\\
&\hspace{2.3cm}-\Gamma_{1}(\Gamma_{1}(\Gamma_{1}(\Phi_1,\Phi_3), \Phi_2), \Phi_4)-\Gamma_{1}(\Gamma_{1}(\Gamma_{1}(\Phi_1,\Phi_3), \Phi_4), \Phi_2)\\
&\hspace{2.3cm}+\Gamma_{1}(\Gamma_{1}(\Gamma_{1}(\Phi_2,\Phi_3), \Phi_1), \Phi_4)+\Gamma_{1}(\Gamma_{1}(\Gamma_{1}(\Phi_2,\Phi_3), \Phi_4), \Phi_1)\\
&\hspace{2.3cm}+\Gamma_{1}(\Gamma_{1}(\Gamma_{1}(\Phi_1,\Phi_4), \Phi_2), \Phi_3)+\Gamma_{1}(\Gamma_{1}(\Gamma_{1}(\Phi_1,\Phi_4), \Phi_3), \Phi_2)\\
&\hspace{2.3cm}+\Gamma_{1}(\Gamma_{1}(\Phi_1, \Phi_3), \Gamma_{1}(\Phi_2, \Phi_4))- \Gamma_{1}(\Gamma_{1}(\Phi_2, \Phi_3), \Gamma_{1}(\Phi_1, \Phi_4))\\
&\hspace{2.3cm}+(\Phi_3''\Phi'_1-\Phi''_1\Phi'_3)(\Phi_4''\Phi'_2-\Phi''_2\Phi'_4)-(\Phi_3''\Phi'_2-\Phi''_2\Phi'_3)(\Phi_4''\Phi'_1-\Phi''_1\Phi'_4)\\
&\hspace{2.3cm}+2(\Phi_2''\Phi'_1-\Phi''_1\Phi'_2)(\Phi_4''\Phi'_3-\Phi''_3\Phi'_4)\Big\}dx\\
&+\frac{1}{4}\int \chi\partial\chi'^2\Big\{-\Gamma_1(\Gamma_1(\Phi_2,\Phi_4), \Phi_1) \Phi'_3-\Gamma_1(\Gamma_1(\Phi_2,\Phi_4), \Phi_3) \Phi'_1\\
&\hspace{2.6cm}-\Gamma_1(\Gamma_1(\Phi_1,\Phi_3), \Phi_2) \Phi'_4-\Gamma_1(\Gamma_1(\Phi_1,\Phi_3), \Phi_4) \Phi'_2\\
&\hspace{2.6cm}+\Gamma_1(\Gamma_1(\Phi_2,\Phi_3), \Phi_1) \Phi'_4+\Gamma_1(\Gamma_1(\Phi_2,\Phi_3), \Phi_4) \Phi'_1\\
&\hspace{2.6cm}+\Gamma_1(\Gamma_1(\Phi_1,\Phi_4), \Phi_2) \Phi'_3+\Gamma_1(\Gamma_1(\Phi_1,\Phi_4), \Phi_3) \Phi'_2\\
&\hspace{2.6cm}+\frac{1}{2}\partial(\Gamma_1(\Phi_1,\Phi_3) \Gamma_1(\Phi_2, \Phi_4))-\frac{1}{2}\partial(\Gamma_1(\Phi_2,\Phi_3) \Gamma_1(\Phi_1, \Phi_4))\Big\}dx\\
&+\frac{1}{4}\int \chi\chi'\partial\chi'\Big\{-\partial(\Phi'_2\Phi'_4 \Phi'_1) \Phi'_3-\partial(\Phi'_2\Phi'_4 \Phi'_3) \Phi'_1-\partial(\Phi'_1\Phi'_3 \Phi'_2) \Phi'_4-\partial(\Phi'_1\Phi'_3 \Phi'_4) \Phi'_2\\
&\hspace{2.7cm}+\partial(\Phi'_2\Phi'_3 \Phi'_1) \Phi'_4+\partial(\Phi'_2\Phi'_3 \Phi'_4) \Phi'_1+\partial(\Phi'_1\Phi'_4 \Phi'_2) \Phi'_3+\partial(\Phi'_1\Phi'_4 \Phi'_3) \Phi'_2\Big\}dx\\
=&\quad\frac{1}{2}\int  \chi''\chi^2\Big\{-\Phi'_2\Phi'_4\Phi''_1\Phi''_3-\Phi'_1\Phi'_3\Phi''_2\Phi''_4+\Phi'_2\Phi'_3\Phi''_1\Phi''_4+\Phi'_1\Phi'_4\Phi''_2 \Phi''_3\Big\}dx. 
\end{split}
\end{equation*}
In the last equality of the above formula, we use the fact that the coefficient of $\chi'$ becomes zero and apply the result in Lemma \ref{lemma8}. This finishes the derivation of the Riemannian curvature in $(\sP_+, \g)$. In addition, 
{
\begin{equation*}
\begin{split}
&\langle \bar \R(\V_{\Phi_1}, \V_{\Phi_2})\V_{\Phi_2}, \V_{\Phi_1}\rangle\\
=&\frac{1}{2}\int \chi''\chi^2\Big\{-\Phi_2'\Phi_1'\Phi_1''\Phi_2''-\Phi_1'\Phi_2'\Phi_2''\Phi_1''+\Phi'_2\Phi'_2\Phi''_1\Phi''_1+\Phi'_1\Phi'_1\Phi''_2\Phi''_2\Big\}dx\\
=&\frac{1}{2}\int \chi''\chi^2 \Big(\Phi''_2\Phi_1'-\Phi_1''\Phi'_2\Big)^2dx .\end{split}
\end{equation*}}
This proves equation \eqref{sec_tensor_exp}. {Note that $Z\geq 0$.} If $\chi(\rho)$ is convex in $\rho$, then $\chi''(\rho)\geq 0$ for any $\rho$. Thus $\bar\K(\V_{\Phi_1}, \V_{\Phi_2})\geq 0$. If $\chi(\rho)$ is concave in $\rho$, then $\chi''(\rho)\leq 0$. Thus $\bar\K(\V_{\Phi_1}, \V_{\Phi_2})\leq 0$. 
\end{proof}
\begin{remark}
{Estimating sectional curvature on density manifolds in high dimensional spaces is an interesting problem. It requires a detailed derivation of the commutator $[\V_{\Phi_1}, \V_{\Phi_2}]$ from Lemma \ref{communtator}, which has been studied in the Wasserstein--2 metric \cite{Lott}. For simplicity of presentation, this paper focuses on the Gamma calculus formulation of curvatures on generalized density manifolds, and we leave quantitative estimates in high dimensional space to future work.}
\end{remark}
\section{Examples}\label{sec5}
In this section, we present physics models induced hydrodynamical density manifolds with geometric quantities, including metrics and sectional curvatures. Three examples of zero range models are studied, including independent particles, simple exclusion processes, and Kipnis-Marchioro-Presutti models.

Consider {a class of zero range models}. The mobility matrix functions are often chosen as $\chi(\rho)=\varphi(\rho)\mathbb{I}$, where $\varphi\in C^{\infty}(\mathbb{R})$, $\varphi(\rho)\geq 0$ and $\mathbb{I}\in\mathbb{R}^{d\times d}$ is an identity matrix. 
In this case, the local Einstein condition satisfies {$D(\rho)=f''(\rho)\cdot \varphi(\rho)$}. Thus, the free energy functional follows 
\begin{equation*}
\mathrm{D}_{f}(\rho, \pi)=\int \Big[f(\rho)-f(\pi)-f'(\pi)(\rho-\pi)\Big]dx. 
\end{equation*}
{Thus, the hydrodynamics \eqref{hydro} satisfies
\begin{equation*}
\begin{split}
\partial_t\rho
=&\nabla\cdot\big(\varphi(\rho)\nabla \frac{\delta}{\delta\rho}\mathrm{D}_f(\rho,\pi)\big).
\end{split}
\end{equation*}}
Both free energy and equation \eqref{hydro} define a Riemannian density manifold $(\sP_+, \g)$. Denote $\Phi_1$, $\Phi_2\in C^{\infty}(\Omega)$, then the Riemannian metric $\g$ is defined as
\begin{equation*}
\g(\V_{\Phi_1}, \V_{\Phi_2})=\int (\nabla\Phi_1, \nabla\Phi_2)\varphi(\rho)dx, 
\end{equation*}
where $\V_{\Phi_i}=-\nabla\cdot(\varphi(\rho)\nabla \Phi_i)$, $i=1,2$. When $\Omega=\mathbb{R}^1$, from Theorem \ref{thm3}, the sectional curvature in $(\sP_+, \g)$ satisfies 
\begin{equation*}
\begin{split}
\bar\K(\V_{\Phi_1}, \V_{\Phi_2})
=&\frac{1}{2Z}\int \varphi''(\rho)\varphi(\rho)^2\big(\Phi''_2\Phi'_1-\Phi''_1\Phi'_2\big)^2dx,
\end{split}
\end{equation*}
where $Z$ is a nonnegative constant defined as $Z=\int |\Phi'_1|^2\varphi(\rho) dx\cdot \int |\Phi'_2|^2\varphi(\rho) dx-(\int \Phi'_1\Phi'_2\varphi(\rho)dx)^2$. 
\begin{example}[Independent particles]
Consider a zero range process with independent particles, i.e. $\varphi(\rho)=\rho\mathbb{I}$, and $D(\rho)=\mathbb{I}$. Denote the free energy with $f(\rho)=\rho\log\rho$, such that 
\begin{equation*}
\mathrm{D}_{f}(\rho, \pi)=\int \Big[\rho\log\rho-\rho\log\pi\Big] dx. 
\end{equation*}
The hydrodynamics \eqref{hydro} satisfies 
\begin{equation*}
\partial_t\rho=-\nabla\cdot(\rho E)+\Delta\rho=\nabla\cdot(\rho\nabla\log\frac{\rho}{\pi}), \quad\textrm{where}\quad E=\nabla\log \pi.
\end{equation*}
{The local Einstein relation satisfies $(\rho\log\rho)''\cdot \rho=\frac{1}{\rho}\cdot\rho=1$.} The free energy and equation \eqref{hydro} define a density manifold $(\sP_+, \g)$, namely Wasserstein-2 metric space \cite{Lafferty, Lott, vil2008}. Thus, the metric $\g$ satisfies 
\begin{equation*}
\g(\V_{\Phi_1}, \V_{\Phi_2})=\int (\nabla\Phi_1, \nabla\Phi_2)\rho dx. 
\end{equation*}
If $\Omega=\mathbb{R}^1$, {both Riemannian and sectional curvatures} of Wasserstein-2 space \cite{am2006,Lafferty,Lott} {are zero. In particular,}
\begin{equation*}
\bar\K(\V_{\Phi_1}, \V_{\Phi_2})=0. 
\end{equation*}
\end{example}
\begin{example}[Simple exclusion]
Consider a simple exclusion process \cite{MFT}, defined by mobility functions $\chi(\rho)=\rho(1-\rho)\mathbb{I}$ and $D(\rho)=\mathbb{I}$. {
We let $\rho\in[0,1]$, such that mobility $\chi(\rho)\geq 0$. 
The operator $-\nabla\cdot(\chi(\rho)\nabla)$ becomes a degenerate elliptic operator at $\rho=0$ and $\rho=1$. 
Here $\rho=0$ and $\rho=1$ act as boundary points of the density manifold. To avoid this degeneracy, we only study the smooth density submanifold
\[
\mathcal{P}_{(0,1)}=\Big\{\rho\in C^\infty(\Omega)~\colon~\int_\Omega \rho dx=1,\quad 0<\rho<1\Big\}.
\]
All derivations below remain valid on the density submanifold.}

Define the free energy with $f(\rho)=\rho\log\rho+(1-\rho)\log(1-\rho)$, such that
\begin{equation*}
\mathrm{D}_{f}(\rho, \pi)=\int \Big[\rho\log\frac{\rho (1-\pi)}{\pi(1-\rho)} +\log\frac{1-\rho}{1-\pi} \Big]dx. 
\end{equation*}
The hydrodynamics \eqref{hydro} satisfies 
\begin{equation*}
\partial_t\rho=-\nabla\cdot(\rho (1-\rho) E)+\Delta\rho=\nabla\cdot(\rho(1-\rho)\nabla \log\frac{\rho (1-\pi)}{(1-\rho)\pi}),\quad \textrm{where}\quad E=\nabla\log \frac{\pi}{1-\pi}.
\end{equation*}
{The local Einstein relation satisfies $(\rho\log\rho+(1-\rho)\log(1-\rho))''\cdot \rho (1-\rho)=\frac{1}{\rho(1-\rho)}\cdot\rho(1-\rho)=1$.} The free energy and equation \eqref{hydro} defines a manifold {$(\mathcal{P}_{(0,1)}, \g)$}, where 
the metric $\g$ satisfies 
\begin{equation*}
\g(\V_{\Phi_1}, \V_{\Phi_2})=\int (\nabla\Phi_1, \nabla\Phi_2)\rho(1-\rho) dx. 
\end{equation*}
If $\Omega=\mathbb{R}^1$, the sectional curvature in {$(\mathcal{P}_{(0,1)}, \g)$} follows 
\begin{equation*}
\bar\K(\V_{\Phi_1}, \V_{\Phi_2})= -\frac{1}{Z}\int \rho^2(1-\rho)^2\big(\Phi''_2\Phi'_1-\Phi''_1\Phi'_2\big)^2dx\leq 0. 
\end{equation*}

\end{example}

\begin{example}[Kipnis-Marchioro-Presutti model]
Consider a Kipnis-Marchioro-Presutti model \cite{KMP} with mobility functions $\chi(\rho)=\rho^2\mathbb{I}$ and $D(\rho)=\mathbb{I}$. The model is defined from heat conduction in a crystal. 
Define the free energy with $f(\rho)=-\log\rho$, such that
\begin{equation*}
\mathrm{D}_{f}(\rho, \pi)=\int \Big[\frac{\rho}{\pi}-\log\frac{\rho}{\pi}-1\Big]dx. 
\end{equation*}
The hydrodynamics \eqref{hydro} satisfies 
\begin{equation*}
\partial_t\rho=-\nabla\cdot(\rho^2 E)+\Delta\rho=\nabla\cdot(\rho^2\nabla (\frac{1}{\pi}-\frac{1}{\rho})),\quad \textrm{where}\quad E=-\nabla\frac{1}{\pi}.
\end{equation*}
{The local Einstein relation satisfies $(-\log\rho)''\cdot \rho^2=\frac{1}{\rho^2}\cdot\rho^2=1$.}
Again, the free energy and equation \eqref{hydro} defines a density manifold {$(\sP_+, \g)$}, where 
the metric $\g$ satisfies 
 \begin{equation*}
\g(\V_{\Phi_1}, \V_{\Phi_2})=\int (\nabla\Phi_1, \nabla\Phi_2)\rho^2 dx. 
\end{equation*}
If $\Omega=\mathbb{R}^1$, the sectional curvature in $(\sP_+, \g)$ follows 
\begin{equation*}
\bar\K(\V_{\Phi_1}, \V_{\Phi_2})= \frac{1}{Z}\int \rho^4\big(\Phi''_2\Phi'_1-\Phi''_1\Phi'_2\big)^2dx\geq 0. 
\end{equation*}
\end{example}
 \section{Discussions}\label{sec6}
 {Our work is motivated by the Otto calculus \cite{otto2001}, from which one studies the Riemannian manifold structure for the Wasserstein-2 space.} This paper develops geometric calculations for hydrodynamical density manifolds arising from Onsager reciprocal relations.
Irreversible processes in complex systems can be formulated as gradient flows on hydrodynamical density manifolds $(\mathcal{P}_+, \g)$, where the metric $\g$ is induced by Onsager response operators. These metrics can be viewed as generalizations of the classical Wasserstein-$2$ metric operator. We derive the Levi--Civita connection, parallel transport, geodesics, and Riemannian and sectional curvatures in these manifolds. In addition, we provide explicit formulas for curvature quantities in one dimensional domains, together with several concrete examples for zero range models, including independent particles, simple exclusion processes, and the Kipnis--Marchioro--Presutti model.  

In non-equilibrium thermodynamics, Onsager response operators have been extensively studied in the MFT~\cite{MFT} and the GENERIC framework~\cite{GEN}. These operators also induce Riemannian metrics on hydrodynamical density manifolds. This paper introduces a corresponding family of infinite dimensional curvature quantities, which we refer to as \emph{macroscopic fluctuation curvatures}. In future work, we will study these curvatures for general physical domains, for example when $\Omega$ is a finite-dimensional Riemannian manifold \cite{Nakamura2024}. In such settings, the macroscopic curvatures interact with the geometric structure of $\Omega$; see Theorem~\ref{theorem2}. Classical examples include Gamma calculus~\cite{BE} on $\Omega$ and Hessian operators for entropy functionals on density manifolds~\cite{C1,LiG2,vil2008}. The geometric framework developed here provides a way to introduce density-dependent or mean-field-type curvatures on the underlying sample space.  {In physical models, one may use the signs of sectional curvatures to distinguish different classes of mobility functions in macroscopic models.
We leave the estimation of these curvature quantities in future work.} Another ongoing direction is to design fast and efficient computational algorithms for hydrodynamics in generalized density manifolds that preserve geometric structures, including free-energy dissipations.
{In applications, we expect that the estimation of macroscopic curvatures will play an important role in the development of stochastic interacting particle systems for machine learning related sampling problems~\cite{LiG5}. For example, we will study accelerated sampling schemes on generalized density manifolds, in which the selection of mobility function, the optimal damping parameters, and the convergence analysis of the scheme depend on the estimations of sectional curvatures and geodesic convexities of free energies.}

\noindent\textbf{Acknowledgements}. {W. Li's work is supported by AFOSR YIP award No. FA9550-23-1-0087, NSF RTG: 2038080, NSF DMS-2245097, and the McCausland Faculty Fellow in University of South Carolina.}


\begin{thebibliography}{10}

\bibitem{Amari2009}
S.~Amari.
\newblock $\alpha$-divergence is unique, belonging to both $f$-divergence and
  Bregman divergence classes.
\newblock {\em IEEE Transactions on Information Theory}, 55(11):4925--4931,
  2009.

\bibitem{am2006}
L.~Ambrosio, N.~Gigli, and G.~Savar{\'e}.
\newblock {\em Gradient flows: in metric spaces and in the space of probability
  measures}.
\newblock Springer Science \& Business Media, 2006.

\bibitem{AyJostLeSchwachhoefer2017}
N.~Ay, J.~Jost, H.~V.~L\^{e}, and L.~Schwachh\"ofer.
\newblock {\em Information Geometry}, volume~64.
\newblock Springer, Cham, 2017.

\bibitem{BE}
D.~Bakry and M.~{\'E}mery.
\newblock Diffusions hypercontractives.
\newblock {\em S{\'e}minaire de
  {{Probabilit{\'e}s XIX}}}, volume 1123, 177--206, 1985.

\bibitem{MFT}
L.~Bertini, A.D.~Sole, D.~Gabrielli, G.~Jona-Lasinio, and C.~Landim. 
\newblock{Macroscopic fluctuation theory.}
\newblock{\em Rev. Mod. Phys.}, 87, 593, 2015. 

\bibitem{C1}
\newblock{J. A. Carrillo, S. Lisini, G. Savare, and D. Slepcev. }
\newblock{Nonlinear mobility continuity equations and generalized displacement convexity.}
\newblock{\em Journal of Functional Analysis} 258(4):1273-1309, {2010}. 


\bibitem{PHC}
P.H.~Chavanis. 
\newblock{Generalized stochastic Fokker-Planck equations.}
\newblock{\em Entropy}, 17(5), 3205-3252, 2015. 

\bibitem{Dean}
D.~Dean.   
\newblock{Langevin equation for the density of a system of interacting Langevin processes.}
\newblock{\em Journal of Physics A: Mathematical and General}, Volume 29, Number 24, 1996. 

\bibitem{O1}
J. Dolbeault, B. Nazaret, and G. Savare. 
\newblock{A new class of transport distances.}
\newblock{\em Calculus of Variations and Partial Differential Equations}, (2):193--231, 2010. 

\bibitem{Gigli2011}
N.~Gigli.
\newblock {\em Second Order Analysis on $(\mathcal{P}_2(M), W_2)$}.
\newblock Memoirs of the American Mathematical Society, Volume 216, Number 1018.
\newblock American Mathematical Society, 2011.


\bibitem{GEN}
M. Grmela and H.C. Ottinger. 
\newblock Dynamics and thermodynamics of complex fluids. I. Development of a general formalism. 
\newblock{\em Phys. Rev. E.}, 56 (6): 6620Ð6632, 1997. 



\bibitem{JKO}
R.~Jordan, D.~Kinderlehrer, and F.~Otto.
\newblock The {{variational formulation}} of the {{Fokker}}--{{Planck equation}}.
\newblock {\em SIAM Journal on Mathematical Analysis}, 29(1):1--17, 1998.

\bibitem{KMP}
C.~Kipnis, C.~Marchioro, and E.~Presutti. 
\newblock{Heat flow in an exactly solvable model.}
\newblock{\em Journal of Statistics Physics}, 27, 65–74, 1982.


\bibitem{Lafferty}
J.D.~Lafferty.
\newblock The {{density manifold}} and {{configuration space quantization}}.
\newblock {\em Transactions of the American Mathematical Society},
  305(2):699--741, 1988.

\bibitem{LiG1}
W.~Li.
\newblock Transport information geometry: Riemannian calculus on probability simplex. 
\newblock {\em Information Geometry}, 5:161--207, 2022.

\bibitem{LiG2}
W.~Li.
\newblock Diffusion hypercontractivity via generalized density manifold. 
\newblock{\em Information Geometry}, 2023. 

\bibitem{LiG3}
W.~Li.
\newblock Hessian metric via transport information geometry. 
\newblock {\em Journal of Mathematical Physics}, 62, 033301, 2021.

\bibitem{LiG4}
W.~Li.
\newblock Transport information Bregman divergences. 
\newblock {\em Information Geometry}, 4, 435--470, 2021.

\bibitem{LiG5}
W.~Li. 
\newblock{Langevin dynamics for the probability of finite step Markov processes.}
\newblock{\em Information Geometry}, 2024. 

\bibitem{LiYing}
W.~Li, and L.~Ying. 
\newblock Hessian transport gradient flows.
\newblock{\em Research in the Mathematical Sciences}, 6, 34, 2019. 

\bibitem{Lott}
J.~Lott.
\newblock Some {{geometric calculations}} on {{Wasserstein space}}.
\newblock {\em Communications in Mathematical Physics}, 277(2):423--437, 2008.

\bibitem{Nakamura2024}
T.~Nakamura.
\newblock Derivation of the invariant free-energy landscape based on Langevin dynamics.
\newblock {\em Physical Review Letters}, 132:137101, 2024.


\bibitem{WD}
M.K. {von Renesse}, and K.-T. Sturm.
\newblock Entropic measure and {{Wasserstein}} diffusion.
\newblock {\em The Annals of Probability}, 37(3):1114--1191, 2009.

\bibitem{Onsager}
L.~Onsager.
\newblock{Reciprocal Relations in Irreversible Processes. I.}
\newblock{\em Phys. Rev.} 37, 1931. 

\bibitem{otto2001}
F.~Otto.
\newblock The geometry of dissipative evolution equations: The porous medium
  equation.
\newblock {\em Communications in Partial Differential Equations},
  26(1-2):101--174, 2001.

\bibitem{vil2008}
C.~Villani.
\newblock {\em Optimal Transport: Old and New}.
\newblock Number 338 in Grundlehren der mathematischen Wissenschaften.
  {Springer}, Berlin, 2009.
  
\bibitem{Zhang2004}
J.~Zhang.
\newblock Divergence function, duality, and convex analysis.
\newblock {\em Neural Computation}, 16(1):159--195, 2004.
    
\end{thebibliography}
\end{document}